\documentclass[acmsmall]{ec20acm}

\AtBeginDocument{%
  \providecommand\BibTeX{{%
    \normalfont B\kern-0.5em{\scshape i\kern-0.25em b}\kern-0.8em\TeX}}}

\copyrightyear{2020}
\acmYear{2020}
\setcopyright{acmlicensed}
\acmConference[EC '20] {Proceedings of the 21st ACM conference on Economics and Computation}{July 13--17, 2020}{Virtual Event, Hungary}
\acmBooktitle{Proceedings of the 21st ACM conference on Economics and Computation (EC '20), July 13--17, 2020, Virtual Event, Hungary}
\acmPrice{15.00}
\acmDOI{10.1145/3391403.3399482}
\acmISBN{978-1-4503-7975-5/20/07}

\usepackage{amsmath}
\usepackage{amsthm}
\usepackage{amssymb}
\usepackage{float}
\usepackage[short]{optidef}
\usepackage{subfiles}
\usepackage{subcaption}
\usepackage{bm}
\usepackage{bbm}
\usepackage{tikz}
\usepackage{xfrac}

\usepackage[shortlabels]{enumitem}
\usepackage[colorinlistoftodos,prependcaption,textsize=tiny]{todonotes}

\newtheorem{theorem}{Theorem}
\newtheorem{proposition}{Proposition}[section]
\newtheorem{lemma}{Lemma}[section]
\newtheorem*{lemma*}{Lemma} 

\DeclareMathOperator{\N}{\text{Normal}}
\DeclareMathOperator{\Pow}{\text{Pareto}}

\DeclareMathOperator{\E}{\mathrm{E}}
\DeclareMathOperator{\Pb}{\mathrm{P}}
\DeclareMathOperator{\1}{\mathrm{1}}
\DeclareMathOperator{\Q}{\mathcal{U}}
\DeclareMathOperator{\V}{\mathcal{V}}
\DeclareMathOperator{\Z}{\mathcal{Z}}

\newcommand\Qgreedy{\Q^{\mathrm{obl}}}
\newcommand\Qdp{\Q^{\mathrm{dp}}}
\newcommand\Qopt{\Q^{\mathrm{opt}}}
\newcommand\Qfair{\Q^{\text{$\gamma$-rule}}}

\newcommand\pxagreedy{\pxa^{\mathrm{obl}}}
\newcommand\pxadp{\pxa^{\mathrm{dp}}}
\newcommand\pxaopt{\pxa^{\mathrm{opt}}}
\newcommand\pxafair{\pxa^{\text{$\gamma$-rule}}}

\newcommand\pxbgreedy{\pxb^{\mathrm{obl}}}
\newcommand\pxbdp{\pxb^{\mathrm{dp}}}
\newcommand\pxbopt{\pxb^{\mathrm{opt}}}

\newcommand\pxggreedy{\pxg^{\mathrm{obl}}}

\newcommand\pxgopt{\pxg^{\mathrm{opt}}}

\DeclarePairedDelimiter\ceil{\lceil}{\rceil}
\DeclarePairedDelimiter\floor{\lfloor}{\rfloor}

\newcommand\esp[1]{\mathrm{E}\left[#1\right]}

\newcommand\sProba[1]{\Pb(#1)}

\newcommand{\w}{W}
\newcommand{\wh}{\hat W}
\newcommand{\wha}{\hat W_A}

\newcommand{\whg}{\hat W_G}

\newcommand{\wi}{W_i}
\newcommand{\whi}{\hat W_i}

\newcommand{\pg}{p_G}
\newcommand{\pa}{p_A}
\newcommand{\pb}{p_B}

\newcommand{\phig}{\frac{1}{\sxg}\phi\left(\frac{w-\hat w}{\sxg}\right)}

\newcommand{\phita}{\frac{1}{\sxa}\phi\left(\frac{w-\twha}{\sxa}\right)}
\newcommand{\phitb}{\frac{1}{\sxb}\phi\left(\frac{w-\twhb}{\sxb}\right)}
\newcommand{\phitg}{\frac{1}{\sxg}\phi\left(\frac{w-\twhg}{\sxg}\right)}
\newcommand{\pq}{p_\w(w)}

\newcommand{\tw}{\theta}
\newcommand{\twh}{\hat \theta}
\newcommand{\twha}{\hat\theta_{A}}
\newcommand{\twhb}{\hat\theta_{B}}
\newcommand{\twhg}{\hat\theta_{G}}
\newcommand{\twhgi}{\hat\theta_{G_i}}

\newcommand{\txg}{\twhg}
\newcommand{\txgdp}{\twhg^{\mathrm{dp}}}
\newcommand{\txggreedy}{\twh^{\mathrm{obl}}}
\newcommand{\txgopt}{\twhg^{\mathrm{opt}}}

\newcommand{\txa}{\twha}
\newcommand{\txb}{\twhb}
\newcommand{\ty}{\tw}

\newcommand{\sxa}{\sigma_{A}}
\newcommand{\sxb}{\sigma_{B}}
\newcommand{\sxg}{\sigma_{G}}

\newcommand{\mq}{\mu_{\w}}
\newcommand{\sq}{\sigma_{\w}}

\newcommand{\pxa}{x_{A}}
\newcommand{\pxb}{x_{B}}
\newcommand{\pxg}{x_{G}}

\newcommand{\pyg}{y_{G}}
\newcommand{\pya}{y_{A}}
\newcommand{\pyb}{y_{B}}
\newcommand{\py}{y}

\newcommand{\ax}{\alpha_1}
\newcommand{\ay}{\alpha_2}

\newcommand{\ff}{\sfrac{4}{5}}

\begin{document}

\title{On Fair Selection in the Presence of Implicit Variance}

\author{Vitalii Emelianov}
\affiliation{
  \institution{Univ. Grenoble Alpes, Inria, CNRS, Grenoble INP, LIG}
  \city{Grenoble}
  \country{France}}
\email{vitalii.emelianov@inria.fr}

\author{Nicolas Gast}
\affiliation{
  \institution{Univ. Grenoble Alpes, Inria, CNRS, Grenoble INP, LIG}
  \city{Grenoble}
  \country{France}}
\email{nicolas.gast@inria.fr}

\author{Krishna P. Gummadi}
\affiliation{
  \institution{Max Planck Institute for Software Systems}
  \city{Saarbrücken}
  \country{Germany}}
\email{gummadi@mpi-sws.org}

\author{Patrick Loiseau}
\affiliation{
  \institution{Univ. Grenoble Alpes, Inria, CNRS, Grenoble INP, LIG}
  \city{Grenoble}
  \country{France}}
\email{patrick.loiseau@inria.fr}


\begin{abstract}

Quota-based fairness mechanisms like the so-called Rooney rule or four-fifths rule are used in selection problems such as hiring or college admission to reduce inequalities based on sensitive demographic attributes (gender, ethnicity, etc.). These mechanisms are often viewed as introducing a trade-off between selection fairness and utility (i.e., the overall quality of the selected candidates). In recent work, however, Kleinberg and Raghavan [\emph{Proc. of ITCS '18}] showed that, in the presence of implicit bias in estimating candidates' quality, the Rooney rule can in fact increase the utility of the selection process (beyond improving its fairness).

We argue that even in the absence of implicit bias, the estimates of candidates' quality from different groups may differ in another fundamental way, namely, in their variance. We term this phenomenon \emph{implicit variance} and we ask: can fairness mechanisms be beneficial to the utility of a selection process in the presence of implicit variance (even in the absence of implicit bias)? To answer this question, we propose a simple model in which candidates have a true latent quality that is drawn from a group-independent normal distribution. To make the selection, a decision maker receives an unbiased estimate of the quality of each candidate, with normal noise, but whose variance depends on the candidate's group. We then compare the utility obtained by imposing a fairness mechanism that we term $\gamma$-rule, which includes demographic parity ($\gamma = 1$) and the four-fifths rule ($\gamma = 0.8$) as special cases, to that of a group-oblivious baseline selection algorithm that simply picks the candidates with the highest estimated quality independently of their group. Our main result shows that the demographic parity mechanism always strictly increases the selection utility, while any other $\gamma$-rule also always increases it weakly. We extend our model to a two-stage selection process where the true quality is observed at the second stage and analyze how our results are changed in that case. We finally discuss multiple extensions of our results, in particular to different distributions of the true latent quality.

\end{abstract}



\maketitle

\section{Introduction}
\label{section: introduction}

\paragraph{\textbf{Discrimination in selection and the role of implicit bias.}}

Many selection problems such as hiring or college admission are subject to discrimination \cite{Bertrand04a}, where the outcomes for certain individuals are negatively correlated with their membership in salient demographic groups defined by attributes like gender, race, ethnicity, sexual orientation or religion. Over the past two decades, implicit bias---that is an unconscious negative perception of the members of certain demographic groups---has been put forward as a key factor in explaining this discrimination \cite{implicit_bias_greenwald06}. While human decision makers are naturally susceptible to implicit bias when assessing candidates, algorithmic decision makers are also vulnerable to implicit biases when the data used to train them or to make decisions was generated by humans.

To mitigate the effects of discrimination on candidates from underrepresented groups, various fairness mechanisms\footnote{These mechanisms are sometimes termed ``positive discrimination'' (e.g., in Germany, France, China, or India) or ``affirmative actions'' (in the USA), often referring to their justification as corrective measures against discrimination suffered in the past by disadvantaged groups. In our work, we analyze the effect of these mechanisms in a particular setting of selection problems (with implicit variance) independently of their motivation, hence we use the more neutral term ``fairness mechanisms.''} are adopted in many domains, either by law or through softer guidelines. For instance, the \emph{Rooney rule} \cite{rooney_rule_collins07} requires that, when hiring for a given position, at least one candidate from the underrepresented group be interviewed. The Rooney rule was initially introduced for hiring American football coaches, but it is increasingly being adopted by many other businesses in particular for hiring top executives \cite{Cavicchia15a,Passariello16a}. Another widely used fairness mechanism is the so-called \emph{$\ff$-rule} \cite{holzer00}, that requires that the selection rate for the underrepresented group be at least 80\% of that for the overrepresented group (otherwise one says that there is adverse impact). This rule is part of the ``Uniform Guidelines On Employee Selection Procedures''\footnote{A set of guidelines jointly adopted by the Equal Employment Opportunity Commission, the Civil Service Commission, the Department of Labor, and the Department of Justice in 1978.}. A stricter version of the $\ff$-rule is the so-called \emph{demographic parity} constraint, which requires the selection rates for all groups to be equal. An overview of these and other fairness mechanisms can be found in \cite{holzer00}.

Fairness mechanisms, however, have been the subject of frequent debates. On one hand, they are believed to promote the inclusion of deserving candidates from underrepresented groups who would have otherwise been excluded in particular due to implicit bias. On the other hand, they are viewed as requiring consideration of candidates from underrepresented groups at the expense of candidates from overrepresented groups, which may potentially decrease the overall utility of the selection process, i.e., the overall quality of selected candidates.

\paragraph{\textbf{Formal analysis of fairness mechanisms in the presence of implicit bias.}}

Perhaps surprisingly, the mathematical analysis of the effect of fairness mechanisms on utility in the context of selection problems was initiated only recently by Kleinberg and Raghavan \cite{kleinberg18} (see also an extension to ranking problems in \cite{celis20}). The authors of \cite{kleinberg18} assume that each candidate $i$ has a true latent quality $\wi$ that comes from a group-independent distribution. They model implicit bias by assuming that the decision maker sees an estimate of the quality $\whi = \wi$ for candidates from the well-represented group and $\whi = \wi / \beta$ for candidates from the underrepresented group, where $\beta > 1$ measures the amount of implicit bias. The factor $\beta$ is unknown (as it is implicit bias) and the decision maker selects candidates by ranking them according to $\whi$. Then Kleinberg and Raghavan \cite{kleinberg18} show that, under a well-defined condition (that roughly qualifies scenarios where the bias is large), the Rooney rule improves in expectation the utility of the selection (measured as the sum of true qualities of candidates selected for interview). This result contradicts conventional wisdom that fairness considerations in a selection process are at odds with the utility of the selection process. Rather, it formalizes the intuition that, in the presence of strong implicit bias (which makes it hard to compare candidates across groups), considering the best candidates across a diverse set of groups not only improves fairness but it also has a positive effect on utility.

\paragraph{\textbf{The phenomenon of implicit variance and its role in discrimination.}}

In this paper, we identify and analyze a fundamentally different source of discrimination in selection problems than implicit bias. Even in the absence of implicit bias in a decision maker's estimate of candidates' quality, the estimates may differ between the different groups in their \emph{variance}---that is, the decision maker's ability to precisely estimate a candidate's quality may depend on the candidate's group. There are at least two main reasons for group-dependent variances in practice. The first arises from \emph{candidates}: different groups of candidates may exhibit different variability when their quality is estimated through a given test. For instance, students of different genders have been observed to show different variability on certain test scores \cite{gender_variability_baye16,O'Dea18a}. The second arises from the \emph{decision makers}: decision makers might have different levels of experience\footnote{Or different amounts of data in case of algorithmic decision making.} judging candidates from different groups and consequently, their ability to precisely assess the quality of candidates belonging to different groups might be different. For instance, when hiring top executives, one may have less experience in evaluating the performance of female candidates because there have been fewer women in those positions in the past (in France for instance, there was only one woman CEO amongst the top-40 companies in 2016-2020). The quality estimate's variance might also change from one decision maker to another. For example, in college admissions, recruiters might be able to judge candidates from schools in their own country more accurately than those from international schools.

We term the above issue with quality estimates `\textbf{\emph{implicit variance}}' as decision makers are often unaware of their group-dependent variances. We posit that implicit variance is an omnipresent and fundamental feature affecting selection problems (including in algorithmic decision making). Indeed, having different variances for the different groups is mostly inevitable and hardly fixable, while estimating these variances can be a difficult task (and variance is indeed ignored in many algorithms).  
In this paper, we model the implicit variance phenomenon by assuming that the decision maker sees of an estimate of the quality of a candidate $\whi$ that is equal to the candidate's true latent quality $\wi$ plus an additive noise whose variance depends on the group of the candidate.\footnote{This noise may be a property of the decision maker getting a noisy perception of the candidate's quality or a property of the candidate (i.e., the variability in the candidate's performance).}  In this situation, a natural baseline decision maker to consider is the \emph{group oblivious} selection algorithm that simply selects the candidates with the highest estimated quality, irrespective of their group, to maximize the selection utility. The group oblivious selection algorithm represents not only a decision maker unaware of the implicit variance in their estimates, but also a decision maker determined to \emph{not use} group information.\footnote{If a decision maker knows the group-dependent variances, then they could use the variances together with group information of the candidates to optimize utility, see the Bayesian-optimal algorithm below.} 

Unfortunately, our analysis shows that in the presence of implicit variance, group oblivious selection can lead to underrepresentation of groups with lower-variance quality estimates compared to groups with higher-variance quality estimates. One natural way to address this representation inequality would be to adopt fairness mechanisms proposed to address discrimination in selection such as the ones discussed above; but this poses the same question that was investigated by Kleinberg and Raghavan \cite{kleinberg18} in the case of implicit bias: \emph{what is the effect of fairness mechanisms on the quality of a selection in the presence of implicit variance?}

\paragraph{\textbf{Our model and overview of our results.}}
To answer this question, we propose a simple model of implicit variance with two groups $A$ and $B$: for each candidate $i$, the decision maker gets a quality estimate $\whi = \wi + \sigma_{G_i} \varepsilon_i$, where $G_i$ is the group to which the candidate belongs and $\varepsilon_i$ is a standard normal random variable. The estimator is unbiased but has a variance $\sigma_{G_i}^2$ that depends on the candidate's group. We assume that the true quality comes from a group-independent distribution---assumed normal in our analytical results. In the one-stage selection problem, the decision maker then selects a fraction $\alpha_1$ (called selection budget) of the candidates.

Using this model, we first observe that, for any selection budget $\alpha_1\neq 1/2$, the group oblivious selection algorithm (our baseline without fairness mechanism) leads to a smaller selection rate---i.e., to underrepresentation---for one of the two groups: the low-variance group if $\alpha_1<1/2$ (the most common case) and the high-variance group if $\alpha_1>1/2$. 
Then, we investigate how the utility of the group oblivious baseline is affected when imposing a fairness mechanism. Specifically, we study a generalization of the $\ff$-rule that we call \emph{$\gamma$-rule}, which imposes that the selection rate for a given group is at least $\gamma$ times that of the other group for some parameter $\gamma \in [0, 1]$. This includes both the $\ff$-rule ($\gamma=0.8$) and demographic parity ($\gamma=1$) as special cases.

Our main result shows that for the one-stage selection problem with any selection budget, beyond giving a more fair representation, demographic parity strictly improves the selection utility (measured by the expected quality of a selected candidate). Moreover, the $\gamma$-rule with any $\gamma < 1$ yields a utility strictly lower than demographic parity but still weakly higher than the group-oblivious algorithm. We then consider a two-stage selection process. There, a pre-selection is first made based on the quality estimates. Then the true quality is observed for each pre-selected candidate and the selection is refined to meet a lower second-stage budget $\alpha_2$. This can model hiring decisions where one first makes a short list based on CV and then refines the selection after interview; or grant selection processes that often happen in two stages with only an extended abstract at first stage and then a full proposal at the second stage. In this two-stage selection process, we show that demographic parity strictly improves the selection quality if the first-stage budget is close enough to the second-stage budget (i.e., one does not interview too many more candidates than slots available) or if the first-stage budget is large; but that it may hurt in between---although much less than the gain in the other regimes. As above, the same holds (but weakly) for the $\gamma$-rule with other values of $\gamma$. Finally, through numerical simulations, we show that our analytical results can be extended in particular to cases where the latent quality distribution is not normal.

Overall, our results show that fairness mechanisms can increase utility in selection problems with implicit variance and freed of any bias. In practical scenarios, one may (at least currently) encounter both implicit bias and implicit variance in decision making, and we do not claim that our main takeaway generalizes to those situations. Similarly, if the baseline decision maker is not group oblivious, fairness mechanisms may not increase utility. 
Finally we remark that in selection problems with implicit variance, the key characteristic of a group is high- \emph{vs} low-variance, not minority \emph{vs} majority. In such problems, the group oblivious algorithm often overrepresents the high-variance group. If the high-variance group corresponds to a minority of candidates, this may seem counter-intuitive. We stress however that ($i$) this corresponds to a group oblivious baseline without bias (which might not match certain practical scenarios); and ($ii$) our model \emph{does not require} that the higher-variance group corresponds to a minority group. We argue that the opposite case is equally interesting in practice (in particular when the implicit variance arises from the candidates).

\paragraph{\textbf{Related works}}

There is an abundant literature on fairness in machine learning, in particular on (one-stage) classification, that tackles the question of how to learn a classifier while enforcing some fairness notion in the outcome \cite{Pedreshi08a,Hardt:2016,Zafar17c,Zafar17a,Chouldechova17a,Corbett-Davies:2017,Lipton18a,Mathioudakis19a}. In this literature, fairness is usually seen as a constraint that reduces the classifier's accuracy and the fairness-accuracy tradeoff is analyzed. In contrast, in our work, we examine selection problems in which fairness can improve utility. Selection also differs from classification by the presence of selection budgets (i.e., maximal number of class-1 predictions), which changes the problem significantly.

The problem of selection is considered in \cite{kleinberg18} under the presence of implicit bias \cite{implicit_bias_greenwald06}. In their work, the authors study the Rooney rule \cite{rooney_rule_collins07} as a fairness mechanism and show that under certain conditions, it improves the quality of selection. An extension of the Rooney rule is studied under a similar model in \cite{celis20}, where the authors investigate the ranking problem (of which the selection problem can be seen as a special case) also in the presence of implicit bias and obtain similar results. In both papers, simple mathematical results expressing conditions under which the Rooney rule improves utility are obtained in the limit regime where the number of candidates is very large. We use the same limit regime in our work but, in contrast, we do not consider that there is implicit bias and introduce instead the notion of implicit variance to capture the difference in precision of the quality estimate for different groups. Although our model can easily be extended to incorporate implicit bias as well, we purposely restrict it to the simplest possible form of implicit variance so as to show its effect on the selection problem independently of bias.
Implicit bias, or simply bias (possibly from an algorithm trained on biased data) in the evaluation of candidates quality is certainly a primary factor of discrimination; but it is also one that may reasonably be fixable through the use of algorithms combined with appropriate debiasing techniques and ground truth data \cite{Raghavan20a} (e.g., by learning fair representations of data \cite{zemel13, locatello19}). 

In our work, we also consider the $\ff$-rule \cite{holzer00} (or rather an extension of it that we call the $\gamma$-rule and that includes demographic parity) rather than the Rooney rule. The main difference between the two is that the $\ff$-rule imposes a constraint on the \emph{fraction} of selected candidates from the underrepresented group whereas the Rooney or its extension in \cite{celis20} imposes a constraint on the \emph{number} of selected candidates from the underrepresented group.

The two aforementioned papers \cite{celis20,kleinberg18} essentially analyze one-stage selection, whereby the utility is the sum of the utilities of selected candidates. In practice, many selection problems are done in two stages where the first-stage selection is refined in a second stage with access to finer information. A few recent papers specifically analyze the two-stage setting. In \cite{kannan19}, a two-stage college admission and hiring procedure is considered. The authors study if certain fair policies---irrelevance of group membership and equal opportunity---can be satisfied. They show that it is possible to satisfy both if the college grades are not reported to the employer, but there are settings where these fairness conditions cannot be satisfied even in isolation.
In \cite{emelianov19}, the authors study an optimal multistage selection. They propose two fairness notions for the multistage setting: local (per stage) and global (final stage) fairness and study their trade-off, the price of local fairness. They show that this price is bounded and is dependent on the timing of the sensitive feature revelation (e.g., gender): the later the sensitive feature is given, the lower the price is.  The model that the authors consider assumes an optimal selection procedure that requires knowledge of the data generation procedure.  In contrast, in our work, we study selection procedures that do not require knowledge of the data distribution and analyze when imposing a fairness mechanism at the first stage leads to an improvement of the utility of the selection after the second stage.

Fairness mechanisms has been a subject of a number of studies in the economic literature, in particular from empirical data. In \cite{coate93}, the authors study whether affirmative actions can remove stereotypes about a particular population. In \cite{aff_action_balafoutas12}, an empirical  evaluation of the influence of affirmative actions in recruiting is performed and it is shown that it can bring quality together with equality. Our work complements those studies through a theoretical model that leads to analytical results on the effect of fairness mechanisms in the presence of implicit variance. The model of observed quality in this paper is similar in spirit to the model of statistical discrimination in \cite{phelps72,Aigner77a}, but in contrast to those works we assume that the decision maker does not know the distribution parameters and the discrimination happens for a different reason.


\section{The Model}
\label{section: model}

We consider the following scenario. A decision maker is given $n$ candidates, out of which a subset is selected. Each candidate $i\in\{1, \cdots, n\}$ is endowed with a true latent quality $\wi$.  The qualities $\wi$ are drawn \emph{i.i.d.} from an underlying probability distribution that is group-independent but unknown to the decision maker.  For our analytical results and unless otherwise explicitly specified, we assume that this distribution is a normal distribution of mean $\mu_W$ and variance $\sigma_W^2>0$. The goal of the decision maker is to maximize the expected quality of the selected candidates: $\esp{\sum_{i\in\text{selection}} \wi}$.

\subsection{The implicit variance model}

We assume that the set of candidates can be partitioned in two groups: group $A$ and group $B$. There are $n_A$ candidates from group $A$ and $n_B=n-n_A$ candidates from group $B$.  We refer to them as $A$-candidates and $B$-candidates. When making the selection decision, the decision maker has access to an unbiased estimator of the true quality. We denote the estimator of the quality of candidate $i$ by $\whi$. We assume that the variance of the estimator depends on the group: for a candidate $i$ that belongs to group $G_i\in\{A,B\}$, its estimated quality is
\begin{equation}
  \label{eq: x1g}
  \whi = \left\{
    \begin{array}{ll}
      \wi + \sigma_{A} \cdot \varepsilon_i & \text{ if $i$ is an $A$-candidate,}\\
      \wi + \sigma_{B} \cdot \varepsilon_i &\text{ if $i$ is a $B$-candidate,}
    \end{array}
    \right.
\end{equation}
where $\varepsilon_i$ is a centered random variable from $\N(0,
1)$---the standard normal distribution, of mean $0$ and variance $1$. The variables $\varepsilon_i$ are assumed \emph{i.i.d.}.

Without loss of generality, in the rest of the paper we assume that $\sxa^2 > \sxb^2$, that is that the quality estimate has higher variance for group A.  We note that none of our results require that $A$ is also the minority group, i.e., that $n_A < n_B$. It is possible to think of scenarios where the minority group has lower variance in cases where the difference in variances arises from the candidates. In the example of students tests scores (see Section~\ref{section: introduction}), for instance, one could potentially observe that males have greater variability in topics in which they are in majority. If the difference in variances arises from the decision-maker and has a statistical nature, however, the minority group (for past selections) will have higher variance due to less data points to build the estimator.  Throughout the paper, we refer to this difference in variance as \emph{implicit variance} because we assume that the decision maker does not know the variance of the estimators: it is an unconscious phenomenon. Also, note that a different estimator, having access to different data, will have different implicit variances for the two groups.
Fig.~\ref{fig: pdf illustration} illustrates the resulting distribution of quality estimates for groups $A$ and $B$ for different distributions of the true latent quality (by abuse of notation, we denote by
$\wha$ a variable that has the same distribution as
$\wi + \sigma_{A} \varepsilon_i$ and similarly for $B$).


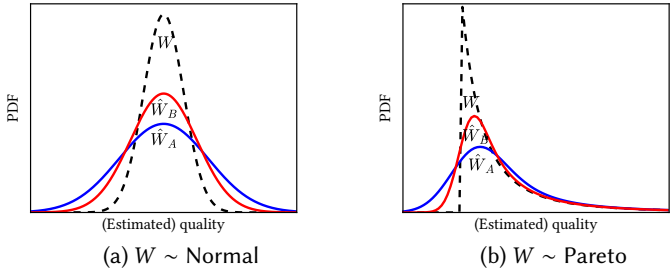
\begin{figure}
  \begin{subfigure}{0.35\textwidth}
  \begin{tikzpicture}
    \node (img)  {\includegraphics[scale=0.45]{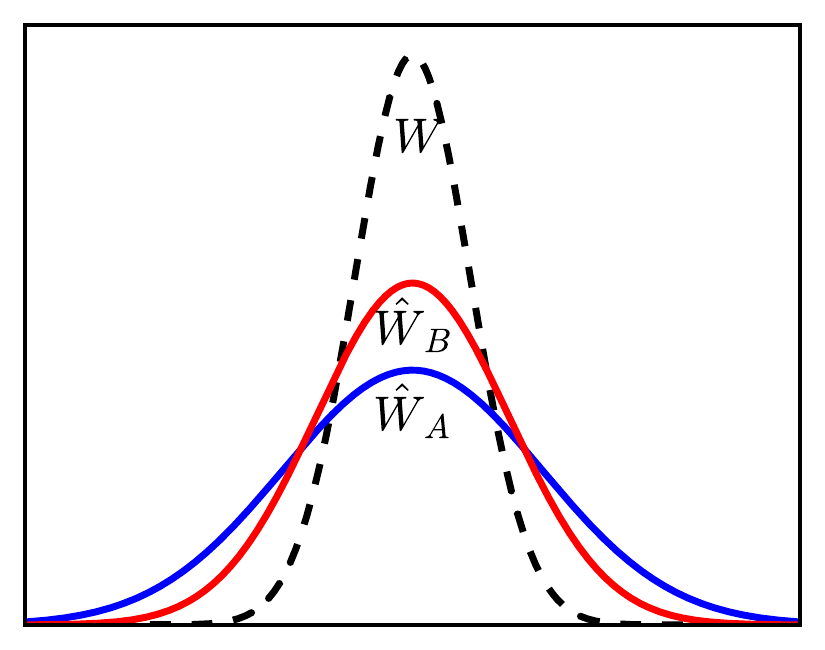}};
    \node[below=of img, node distance=0cm, yshift=1.3cm] {\tiny (Estimated) quality};
    \node[left=of img, node distance=0cm, rotate=90, anchor=center,yshift=-1cm] {\tiny PDF};
   \end{tikzpicture}
   \vspace{-.2cm}\caption{$\w \sim \N$}
  \end{subfigure}
  \begin{subfigure}{0.35\textwidth}
      \begin{tikzpicture}
        \node (img)  {\includegraphics[scale=0.45]{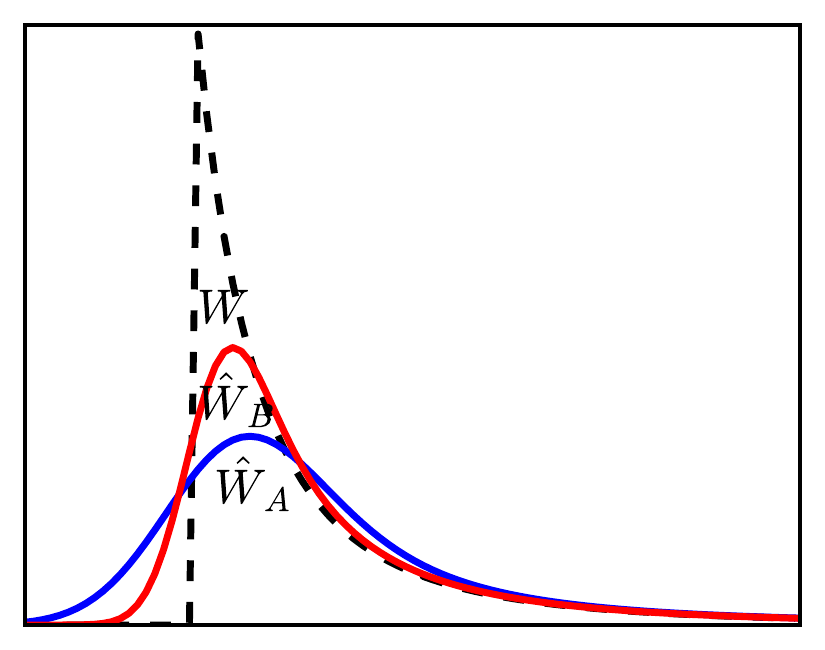}};
        \node[below=of img, node distance=0cm, yshift=1.3cm] {\tiny (Estimated) quality};
        \node[left=of img, node distance=0cm, rotate=90, anchor=center,yshift=-1cm] {\tiny PDF};
      \end{tikzpicture}
      \vspace{-.2cm}\caption{$\w \sim \Pow$}
    \end{subfigure}
\vspace{-.4cm}\caption{Probability density function of the true latent quality $\w$
  and the estimated quality $\wh$.}
\vspace{-.4cm}
\label{fig: pdf illustration}
\end{figure}

\subsection{One-stage and two-stage selection problems}

Candidates are selected in a one-stage or two-stage process.  In the
first stage, for each of the candidate $i$, the decision maker
observes the quality estimate $\whi$ as well as its group
$G_i\in\{A,B\}$. They then select $m_1$ candidates out of those
$n$. The selection process stops here if the selection has only one
stage. 
If the process has two stages, the decision maker observes the true
quality $\wi$ of the $m_1$ candidates that were selected at first
stage. They select the $m_2$ candidates having the largest $\wi$ among
those $m_1$ candidates.  We remark that when $m_1=m_2$, the two-stage
selection process reduces to the one-stage selection process. 

%

    

\subsection{Selection algorithms}
\label{sec:selection_algorithms}

The goal of the decision maker is to maximize the expected quality of the selected candidates. Since the true quality $\wi$ is observed at second stage, the natural selection at second stage is to sort the candidates according to their quality $\wi$ and to select the best $m_2$.  For the first stage, the situation is different because the decision maker only has access to an estimator of the quality $\whi$ whose variance depends on the candidate's group. In this paper, we compare the following first-stage selection algorithms to show that fairness mechanisms can improve the selection quality.

\paragraph{Group Oblivious}
As $\whi$ is an unbiased estimator of $\wi$ whose variance is unknown to the decision maker, the most natural first-stage selection rule is to sort the candidates according to $\whi$ irrespective of their group and to keep the best $m_1$. We call this the \emph{group oblivious} selection. This selection algorithm might be seen as a fair treatment because the selection does not use the group label. Yet, because of our implicit variance model, this might lead to discriminate against high-variance or low-variance groups. We will discuss that in Lemma~\ref{lemma: selections for dp and greedy 1 stage}.

\paragraph{$\ff$-rule and $\gamma$-rule} Let us denote by $\pxa$ (and
$\pxb$) the fraction of the $A$-candidates (and $B$-candidates) that
are selected at first stage.  The group oblivious selection might
favor one group or the other, that is $\pxa\gg\pxb$ or $\pxb\gg\pxa$.
To mitigate the inequality, the decision maker can introduce selection
quotas. One example is the well studied Rooney rule \cite{kleinberg18}
that states that at least one candidate from each group should be
chosen. Another example is the four-fifths rule \cite{holzer00} that
imposes that $\pxa \ge \frac45 \pxb$ and $\pxb \ge \frac45\pxa$. 
In this paper, we consider a generalization of this rule that is
parametrized by $\gamma \in [0, 1]$. We say that a selection is $\gamma$-fair if
\begin{align}
  \pxa \ge \gamma \pxb \qquad\text{ and }\qquad
  \pxb \ge \gamma \pxa.
  \label{eq:beta-fair}
\end{align}
In practice, this is easily done by sorting $A$- and $B$-candidates
separately.  Indeed, by using that $\pxa \ge \gamma \pxb$ and
$\pxb \ge \gamma \pxa$, the total number of candidates selected at
first stage satisfies
\begin{align*}
  m_1=\pxa n_A+\pxb n_B &\le \pxa n_A + \frac{\pxa}{\gamma}n_B
                          = \pxa n_A\frac{n_A\gamma +
                          n_B}{n_A\gamma}, %
                          \quad \textrm{ and }\quad 
  m_1
                          \le \pxb n_B\frac{n_B\gamma + n_A}{n_B\gamma}.
\end{align*}
This means that to satisfy~\eqref{eq:beta-fair}, the selection
algorithm picks the best estimated
$\ceil{m_1\gamma n_A/(n_B+\gamma n_A)}$ $A$-candidates and  the
best estimated $\ceil{m_1\gamma n_B/(n_A+\gamma n_B)}$ $B$-candidates. Then, the remaining positions are filled with the best estimated among the remaining candidates, irrespective of their group.
Note that when $\gamma=0$, the $\gamma$-fair algorithm
reduces to the group oblivious algorithm. 

\paragraph{Demographic parity} When $\gamma=1$, the
$\gamma$-fair algorithm corresponds to the classical notion of
\emph{demographic parity} \cite{Zafar17a} that mandates that the selection rates be equal across different groups.  Note that
because $n_A$, $n_B$ and $m_1$ are integer variables, it might be
impossible to satisfy the constraints in~\eqref{eq:beta-fair}
when $\gamma$ is too close to $1$. In such a case, we say that an
algorithm is $\gamma$-fair if the constraint \eqref{eq:beta-fair} is
satisfied up to one candidate.

\paragraph{Bayesian optimal} We will compare the performance of the
above selection algorithms with the performance of what we call a
\emph{Bayesian-optimal} selection algorithm. This algorithm is an
idealized selection algorithm that knows all the parameters of the
problem (the quality distribution and the variances $\sxa^2$ and
$\sxb^2$) and chooses the candidates in order to maximize the expected
quality at final stage.  Recall that our model assumes that the
decision maker does not know the quality distribution $\w$ nor the
variances of the estimator $\sxa^2$ and $\sxb^2$. Hence, this
Bayesian-optimal algorithm is not implementable in practice; we will
 use it as an upper bound of what could be achieved.

Note that while the previous algorithms are oblivious to what happens
in the second stage, the Bayesian-optimal selection does depend on the
fraction of candidates selected at second stage, $\alpha_2$. 

\subsection{Simplification of the selection problem for large $n$ and $m$}

In the remainder of the paper, we study the selection problem when the
number of candidates is large.  That is, we assume that there exist fixed
fractions $0\le\pa\le1$ and $0<\alpha_2\le\alpha_1\le1$ such that 
\begin{align*}
  n_A = \floor{p_An}\qquad m_1 = \floor{\alpha_1n} \qquad m_2 =
  \floor{\alpha_2n},
\end{align*}
and let $n$ grow. 
Our theoretical results are obtained in the limit where $n$ goes to
infinity (similarly to \cite{kleinberg18,celis20}). In Section~\ref{ssec:approximation} we will show
numerically that our results for $n=\infty$ continue to hold for
finite selection sizes. Note that $p_A$ represents the fraction of $A$-candidates in the population while $\alpha_1$ and $\alpha_2$ represent the global selection ratios (or budgets) at first and second stage respectively.

As we prove below, characterizing the performance of a selection
problem is simpler when the number of candidates $n$ is infinite
because a selection algorithm is characterized by three selection
thresholds $\twha$, $\twhb$ and $\tw$ as follows. The thresholds
$\twhg$'s correspond to the first stage selection and the threshold
$\tw$ corresponds to the second stage selection: if $i$ is a
$G_i$-candidate, they will be selected at first stage if
$\wh_i \ge \twhgi$. They will pass both stages if $\wh_i \ge \twhgi$
and $\w_i \ge \tw$. For given thresholds $\twha$, $\twhb$ and $\tw$,
we denote the expected utility of the corresponding selection by
$\V(\twha, \twhb, \tw)$:
\begin{align*}
  \V(\twha, \twhb, \tw) = %
  \esp{\w_i\,|\, \wh_i \geq \twhgi, \w_i \geq \tw}.
\end{align*}
For these given thresholds $\twha, \twhb, \tw$, the fractions of
selected candidates are $\Pb(\wh_i\ge\twhgi)$ after the first stage,
and $\Pb(\w_i\ge\tw , \wh_i\ge\twhgi)$ after the second stage.  Using
the above definition, we denote by $\Q(\pxa)$ the expected utility of
a threshold-type selection algorithm that selects $A$-candidates with
probability $\pxa$ at the first stage and that satisfies the selection size constraints in
expectation:
\begin{align}
  \label{eq: u}
  \Q(\pxa) &=\V(\twha, \twhb, \tw),\text{ where $\twha, \twhb, \tw$
             are such that}
             \left\{\begin{array}{l}
                      \Pb(\whi \ge \twha\,|\,G_i=A) = \pxa,\\
                      \Pb(\whi \ge \twhgi) = \ax, \\
                      \Pb(\whi \ge \twhgi, \wi \ge \tw) = \ay.
                    \end{array}
  \right.
\end{align}
Note that combining the first two constraints in \eqref{eq: u} immediately gives that such an algorithm selects $B$-candidates with probability $\pxb = (\ax - x_A p_a)/(1-p_A)$. Hence it is sufficient to describe the algorithm with $x_A$.


The above definition of expected quality is not directly applicable to
the selection algorithms presented in
Section~\ref{sec:selection_algorithms} because those algorithms are
defined neither in terms of fraction of selected candidate nor in
terms of thresholds. In fact, for a given selection algorithm, the
fractions of selected $A$- and $B$-candidates depend on the
realizations of the random variables representing the quality ($\wi$)
and the estimated quality ($\whi$). As a result, these fractions ($\pxa$
and $\pxb$) are random variables. For instance, if because of
randomness the $A$-candidates are evaluated much worse than
the $B$-candidates, then $\pxa$ will be $0$ for the group
oblivious algorithm.
The following proposition shows that when the population is large,
these random fluctuations disappear. It shows that, when $n$ is large,
the performance of the various algorithms are simply characterized by
$\pxa$.
\begin{proposition}
  \label{th:n_infinity}
  For any problem parameters and any of the first stage selection
  algorithms presented in Section~\ref{sec:selection_algorithms},
  \begin{enumerate}
  \item there exists a deterministic fraction $x_A\in[0,1]$ such that
    the fraction of $A$-candidates that are selected by the algorithm
    converges (in probability) to $x_A$ as $n$ grows;
  \item there exist deterministic thresholds $\twha, \twhb, \tw$ such
    that the expected utility of this algorithm converges to
    $\V(\twha, \twhb, \tw)$.
  \end{enumerate}
\end{proposition}
\begin{proof}[Sketch of Proof]
  The above result is essentially a direct consequence of the law of
  large numbers.  By the Glivenko-Cantelli theorem, the empirical
  distribution of the estimated qualities of the $G$-candidates
  converges to the distribution of $\whg$ as $n\to\infty$. This shows
  that taking the best $\floor{n\pa\pxa}$ $A$-candidates or taking all
  $A$-candidates above the $\pxa$-quantile of the distribution $\wha$
  is asymptotically equivalent as $n\to\infty$. The same argument can
  be used to show that second stage selection is asymptotically
  equivalent to selecting all candidates above a given threshold.
\end{proof}


In what follows, we will study directly the model when the population
$n$ is large.  We denote by respectively $\pxagreedy$, $\pxafair$,
$\pxadp$, and $\pxaopt$ the asymptotic fraction of $A$-candidates that
are selected at first stage for the group oblivious, the $\gamma$-fair,
the demographic parity and the Bayesian-optimal algorithms. Moreover, we will
denote the expected performance of the various algorithms by
\begin{align*}
  \Qgreedy = \Q(\pxagreedy); %
  \qquad \Qdp = \Q(\pxadp); %
  \qquad \Qfair = \Q(\pxafair); %
  \qquad \Qopt = \Q(\pxaopt). %
\end{align*}
For a finite $n$, characterizing precisely the utility of an
algorithm like group oblivious is computationally difficult due to the
correlations between the selection of the different agents.
Proposition~\ref{th:n_infinity} allows us to greatly simplify the
study of the performance of the various heuristics because the
function $\Q$, defined in Equation~\eqref{eq: u}, depends only on one
parameter $\pxa$, and is simpler to characterize than the expectation
over a finite number of candidates $n$.

\paragraph*{\textbf{Summary of the notation}}

To simplify the exposition, and since candidates are interchangeable,
in the remainder of the paper, we will omit the subscript $i$ and
write directly $\w$ and $\wh$ for the quality and the estimated quality
of a given candidate.  Our notation is summarized in
Table~\ref{table: notation}.



\begin{table}[ht]
  \caption{Summary of notation.}
  \label{table: notation}
  \vspace{-.2cm}  \begin{tabular}{ll}
    \toprule
    $\pg$ & fraction of $G$-candidates, $n_G/n$, for $G\in\{A,B\}$\\
    $\sq^2$ & variance of latent quality $\w$\\
    $\sxg^2$ & implicit variance of estimated quality $\wh$ given
               group $G\in\{A,B\}$\\
    \midrule
    $\pxg$ & fraction of $G$-candidates that are selected at first stage\\
    $\twhg$ & threshold for $G$-candidates at first stage\\
    $\tw$ & threshold at second stage\\
    $\ax$, $\ay$ & fractions of candidates selected at stages 1 and 2 (or budgets)\\
    $\Q$ & expected selection quality\\
    \midrule
    $\phi$, $\Phi$, $\Phi^c$, $\Phi^{-1}$ & pdf, cdf, complementary cdf and quantile of $\N(0,1)$\\
    \bottomrule
\end{tabular}
\end{table}





\section{One-Stage Selection}
\label{section: one stage}

We start our study with the simplest case: one-stage selection.  In this
setting $\ay=\ax$, which means that the decision maker observes the values
of $\wh$ and makes a selection that is final, i.e., no further
subselection is performed after observing the exact values of
$\w$. This type of selection is the most commonly studied in the
related work, see for instance \cite{kleinberg18,celis20}. 

In this section, we compare the first-stage selection algorithms
that we introduced in Section~\ref{sec:selection_algorithms}. The main
result of this section is that using any $\gamma$-fair algorithm
increases the performance compared to using the group-oblivious
algorithm. 
To show this, we start by describing key properties of the one-stage algorithms in Section \ref{ssec:one-stage:policies}. The main result is then
stated and proven in Section \ref{ssec:one-stage:improve}. 

\subsection{Behavior of the different first-stage selection algorithms}
\label{ssec:one-stage:policies}

\paragraph*{Group oblivious} This algorithm sorts the candidates
according to the estimates $\wh_i$'s and selects the candidates having
the highest estimates, regardless of their group. The decision maker
does not distinguish between $A$- and $B$-candidates and
treats them equally.  This corresponds to applying the same threshold
for the two groups, i.e., $\twha = \twhb$.  Due to different variances
of estimation, this might lead to selecting more people from one group
or the other.

Recall that we assume that $\sxa>\sxb$. This means that the distribution of $\wh$ has more extremes for an $A$-candidate.  Thus, if the selection size is small, more $A$-candidates will be selected compared to $B$-candidates because the probability to estimate an $A$-candidate as a ``genius'' is higher than for $B$-candidates. In contrast, if the selection size is large, the chance of estimating an $A$-candidate as bad is larger than for $B$-candidates, in which case the decision maker selects a lower fraction of $A$-candidates. This can be formally stated as follows.
\begin{lemma}[Group oblivious Selection]
  \label{lemma: selections for dp and greedy 1 stage}
  When using the group oblivious selection algorithm, the fractions
  $\pxggreedy$ of selected candidates from each group satisfy:
    \begin{enumerate}
    \item if $\ax < 1/2$, then $\pxagreedy > \pxbgreedy$;
    \item if $\ax > 1/2$, then $\pxagreedy < \pxbgreedy$;
    \item if $\ax = 1/2$, then $\pxagreedy = \pxbgreedy=1/2$.
    \end{enumerate}
\end{lemma}
\begin{proof}[Proof Sketch]
  The selection fraction $\pxg$ is by definition the probability that
  a $G$-candidate has a value of $\wh$ larger than $\twhg$. As
  $\sxa>\sxb$, one has
  $\sProba{\wh\ge \mu_W+x|A} > \sProba{\wh\ge \mu_W+x|B}$ and
  $\sProba{\wh\ge \mu_W-x|A} < \sProba{\wh\ge \mu_W-x|B}$ for any
  $x>0$.  This implies that $\pxagreedy{} > \pxbgreedy{}$ if $\alpha_1 < 1/2$ and
  $\pxagreedy{} < \pxbgreedy{}$ if $\alpha_1>1/2$.  Since the distribution of quality
  $\w$ is symmetric, then one has $\pxagreedy=\pxbgreedy=1/2$ when
  $\alpha_1=1/2$. The proof is detailed in Appendix~\ref{proof: selections for dp and greedy 1 stage}.
\end{proof}

\paragraph*{$\gamma$-rule} An algorithm that satisfies the $\gamma$-rule
is a variant of the group-oblivious algorithm that guarantees a minimum
selection rate for each group.  A straightforward computation shows
that an algorithm satisfies both the $\gamma$-rule conditions of~\eqref{eq:beta-fair} and the size constraint
$\pxa\pa+\pxb(1-\pa)=\alpha_1$ if and only if
$\pxa\in[\frac{\alpha_1}{p_A+p_B/\gamma},
\frac{\alpha_1}{p_A+p_B\gamma}]$. By our definition of
Section~\ref{sec:selection_algorithms}, the $\gamma$-rule selection algorithm is the algorithm that
is the closest to the group-oblivious algorithm while respecting the $\gamma$-rule. Thus, the fraction of $A$-candidates  selected at
first stage of the $\gamma$-rule selection algorithm is
\begin{align*}
  \pxafair = \min\left(\frac{\alpha_1}{p_A+p_B\gamma},
  \max\left(\pxagreedy,\frac{\alpha_1}{p_A+p_B/\gamma}\right)\right).
\end{align*}

\paragraph*{Demographic parity} Contrary to the group oblivious selection algorithm, this algorithm applies different thresholds for the two groups to preserve the demographics. Since the selection size is fixed to $\ax$ and both selection fractions are equal, then $\pxadp = \pxbdp=\ax$.  Thus, if the selection size $\ax$ is small, the demographic parity selection algorithm increases the threshold $\twha$ for $A$-candidates by removing extremes and filling other places by $B$-candidates. If the selection size is large, it lowers the threshold for $A$-candidates by supporting the ones that were removed from the selection.

Formally, using the properties of normal distributions, we can write that for a fixed $\ax$, a $G$-candidate with estimate $\wh$ is selected if $\wh\ge \sqrt{\sq^2 + \sxg^2}\Phi^{-1}(1-\ax) + \mq$. Recall that in our regime with $n=\infty$, the thresholds $\twha$ and $\twhb$ are the estimates of quality of the worst $A$- and $B$-candidates that are selected. Then, with the demographic parity selection algorithm, we have $(\twha-\mu_W)/(\twhb-\mu_W)=\sqrt{\sxa^2 +\sq^2}/\sqrt{\sxb^2 +\sq^2}$.

\paragraph*{Bayesian-optimal selection} Since $(\w, \wh)$ is a bivariate
normal, then using the property of conditional expectation, the
expected quality of candidate given its estimate is
\begin{equation}
  \E(\w|\wh=\hat w) = \frac{\sq^2}{\sxg^2 + \sq^2}\hat w + \left(1 - \frac{\sq^2}{\sxg^2 + \sq^2} \right) \mq.
\label{eq:expected w given w hat}
\end{equation}
The above expression corresponds to the expectation of the posterior
distribution of $\w$ given $\wh$. This expectation is $\wh$ if $\sigma_G^2=0$ (i.e., there is no noise)
and converges to $\mu_W$ when $\sigma^2_G\to\infty$. 

In one-stage selection, the Bayesian-optimal selection algorithm picks the candidates with the highest posterior quality expectation \eqref{eq:expected w given w hat}. This means that the thresholds $\twhg$, which are the estimations of quality of the worst selected $A$- and $B$-candidates, satisfy $(\twha-\mu_W)/(\twhb-\mu_W)=(\sxa^2 +\sq^2)/(\sxb^2 +\sq^2)$.  This shows that when the number of candidates to be selected is small ($\alpha_1<1/2$), the Bayesian-optimal selection is even more conservative than the demographic parity selection: it imposes more strict constraints on the high variance group $A$ if selection is for candidates having estimated qualities higher than the median. We state the results in the following lemma.
 
\begin{lemma}[Bayesian-optimal Selection]
  \label{lemma: optimal selection}
  When using the Bayesian-optimal one-stage selection algorithm, the fraction
  of $A$- and $B$-candidates selected, $\pxaopt$ and
  $\pxbopt$, satisfy
  \begin{enumerate}
  \item if $\ax < 1/2$, then $\pxaopt < \pxbopt$;
  \item if $\ax > 1/2$, then $\pxaopt > \pxbopt$;
  \item if $\ax = 1/2$, then $\pxaopt = \pxbopt=1/2$.
  \end{enumerate}
\end{lemma}
\begin{proof}[Proof Sketch]
  To study the Bayesian-optimal selection, we express the utility $\Q$ as a function of $\pxa$. The key difficulty is to compute the first and second derivatives of $\Q$ with respect to $\pxa$. Then, using Harris inequality \cite{fortuin71}, we show that $\Q$ is strictly concave. Hence the root of the equation $\frac{d\Q}{d\pxa} = 0$ gives the optimal threshold $\twhg$.  The expression for the thresholds allows us to compare the values of $\pxaopt$ and $\pxbopt$ depending on selection size $\ax$. The full proof is given in Appendix~\ref{proof: optimal selection}.
\end{proof}

The behavior of these algorithms is illustrated in Fig.~\ref{fig: selection one stage}. We plot the result for a small selection size ($\alpha_1<0.5$); the situation for $\alpha_1>0.5$ is symmetric. We observe that when the selection size is small, the group oblivious algorithm selects many more $A$-candidates than $B$-candidates. In order to select the same fraction of candidates from both groups, the demographic parity selection algorithm uses a higher threshold for the high-variance candidates. The $\gamma$-rule selection algorithm (for $\gamma=0.8$, which corresponds to the $\ff$ rule) is between the group oblivious and demographic parity selection algorithms. We also observe that the Bayesian-optimal selection algorithm is even more conservative and selects fewer $A$-candidates than the other algorithms.


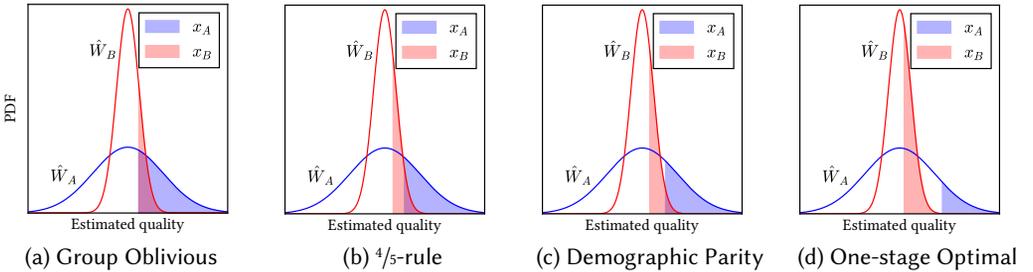
\begin{figure}[ht]
  \begin{subfigure}{0.24\textwidth}
  \begin{tikzpicture}
    \node (img)  {\includegraphics[scale=0.45]{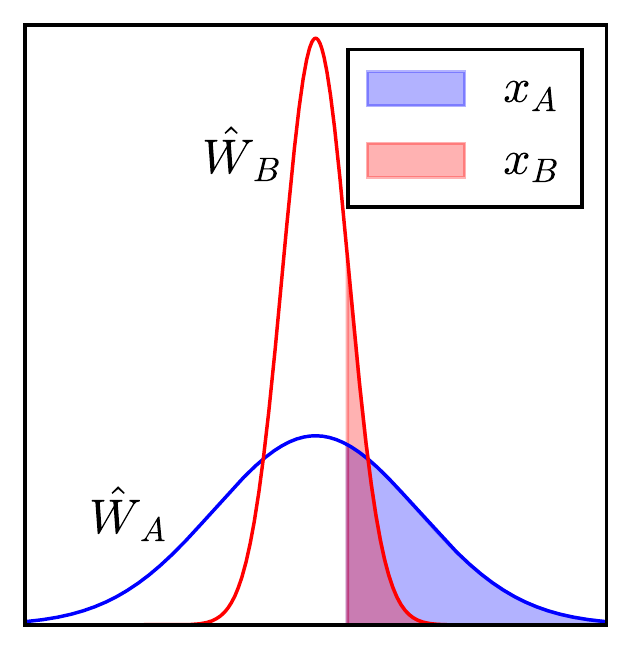}};
    \node[below=of img, node distance=0cm, yshift=1.3cm] {\tiny Estimated quality};
    \node[left=of img, node distance=0cm, rotate=90, anchor=center,yshift=-1cm] {\tiny PDF};
   \end{tikzpicture}
   \vspace{-.6cm}\caption{Group Oblivious} 
  \end{subfigure}
\hspace{0.1cm}
\begin{subfigure}{0.24\textwidth}
  \begin{tikzpicture}
    \node (img)  {\includegraphics[scale=0.45]{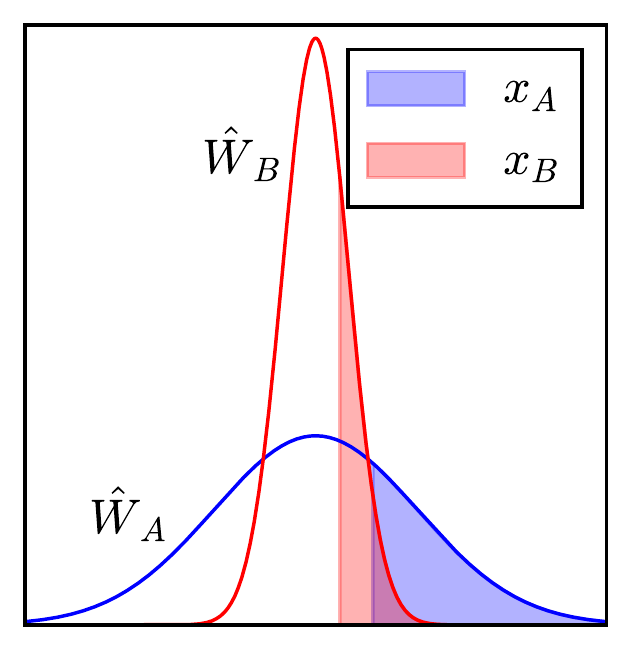}};
    \node[below=of img, node distance=0cm, yshift=1.3cm] {\tiny Estimated quality};
  \end{tikzpicture}
  \vspace{-.2cm}\caption{$\ff$-rule} 
\end{subfigure}
  \begin{subfigure}{0.24\textwidth}
    \begin{tikzpicture}
      \node (img)  {\includegraphics[scale=0.45]{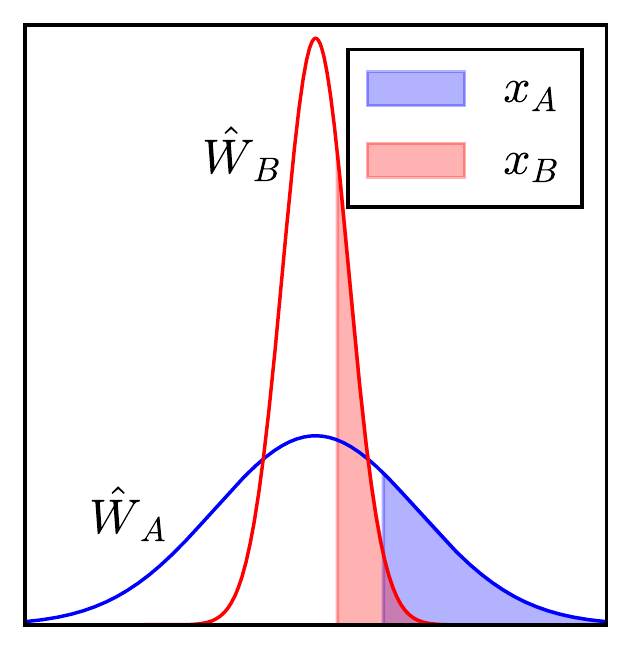}};
      \node[below=of img, node distance=0cm, yshift=1.3cm] {\tiny Estimated quality};
    \end{tikzpicture}
    \vspace{-.2cm}\caption{Demographic Parity} 
  \end{subfigure}
    \begin{subfigure}{0.24\textwidth}
      \begin{tikzpicture}
        \node (img)  {\includegraphics[scale=0.45]{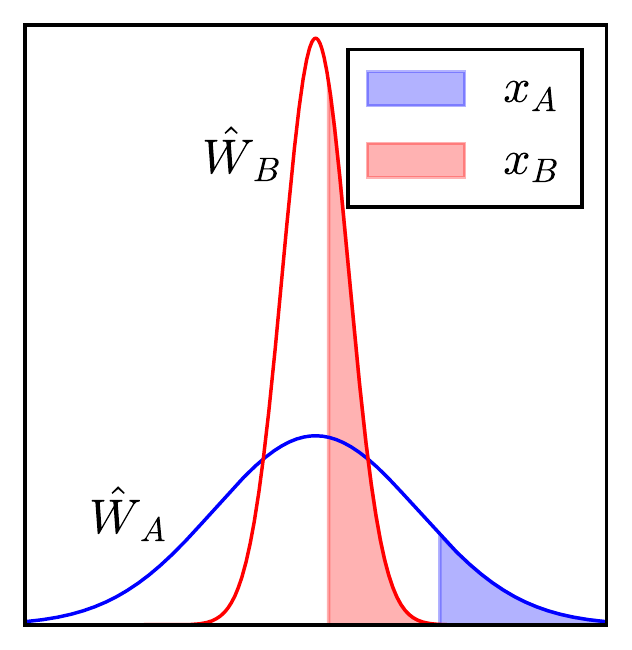}};
        \node[below=of img, node distance=0cm, yshift=1.3cm] {\tiny Estimated quality};
      \end{tikzpicture}
      \vspace{-.2cm}\caption{One-stage Optimal} 
      \end{subfigure}
\vspace{-.2cm}\caption{Illustration of the fraction of selected candidates at first-stage ($\alpha_1<0.5$).}
      \label{fig: selection one stage}
\end{figure}

In Fig.~\ref{fig: utility 1 stage}, we plot the expected utility $\Q(\pxa)$ as a function of the fraction of selected $A$-candidates. We compare different values of selection sizes $\alpha_1\in\{0.15,0.35,0.6,0.8\}$. As stated in the proof of Lemma~\ref{lemma: optimal selection}, this function is concave and its maximum is attained in $\pxaopt$. In this figure, we also plot the $\gamma$-fair regions for $\gamma = 0.8$ ($\ff$ rule). To satisfy the $\gamma$-rule \eqref{eq:beta-fair}, the fraction $\pxagreedy$ selected by the group-oblivious selection algorithm should suffer a correction such that the corresponding selection fraction will lie on the boundary of the $\gamma$-region, if it is outside of it. 
 
\begin{figure}[ht]
  \centering
  \begin{subfigure}{.24\textwidth}
    \includegraphics[width=\linewidth]{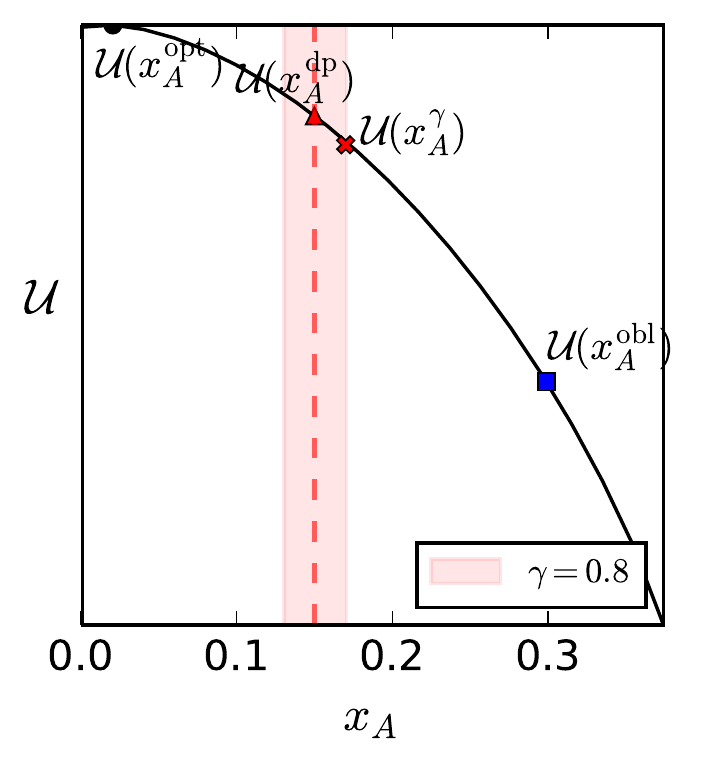} 
    \vspace{-.6cm}\caption{$\alpha_1=0.15$}
    \label{fig: utility 1 stage_a}
  \end{subfigure}
  \begin{subfigure}{.24\textwidth}
    \includegraphics[width=\linewidth]{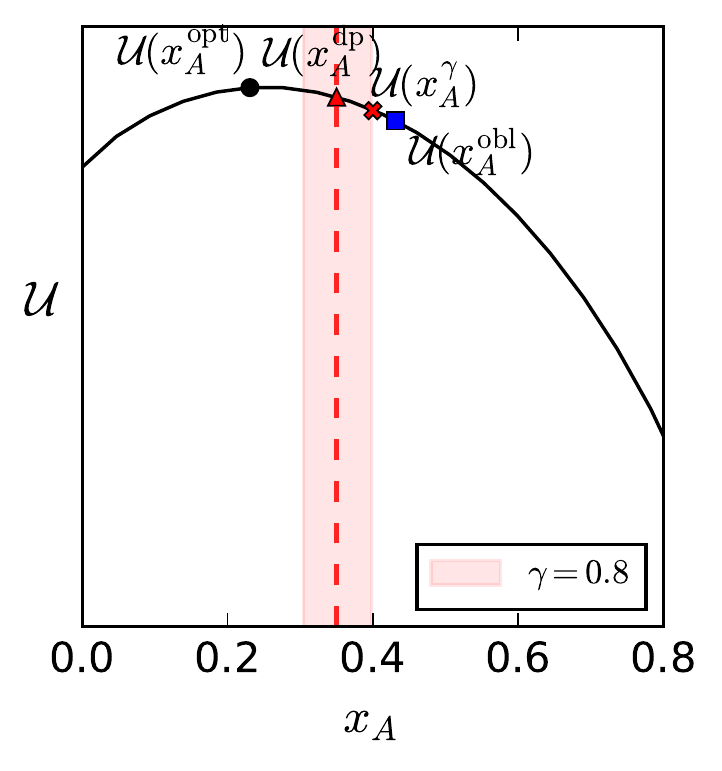} 
    \vspace{-.6cm}\caption{$\alpha_1=0.35$}
    \label{fig: utility 1 stage_b}
  \end{subfigure}
  \begin{subfigure}{.24\textwidth}
    \includegraphics[width=\linewidth]{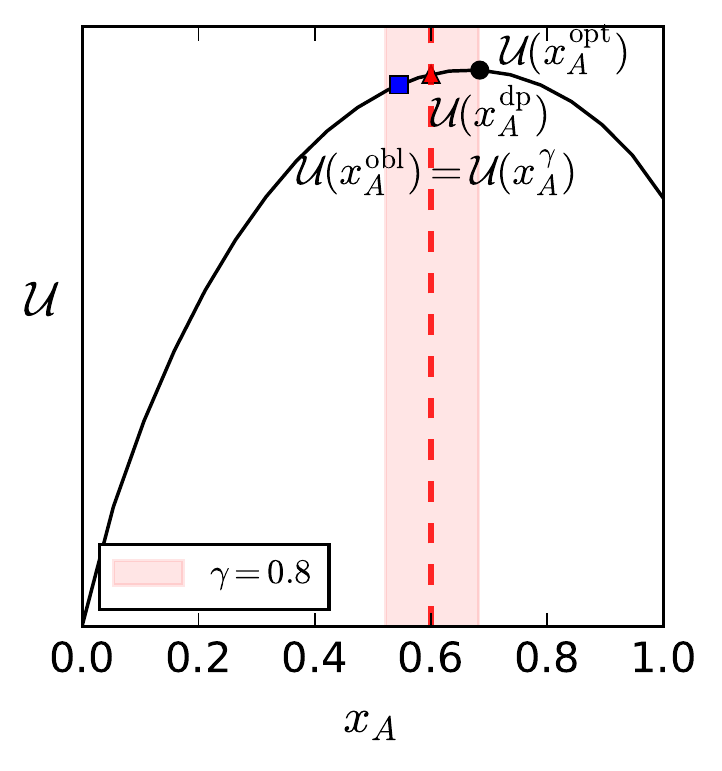}
    \vspace{-.6cm}\caption{$\alpha_1=0.6$}
    \label{fig: utility 1 stage_c}
  \end{subfigure}
  \begin{subfigure}{.24\textwidth}
    \includegraphics[width=\linewidth]{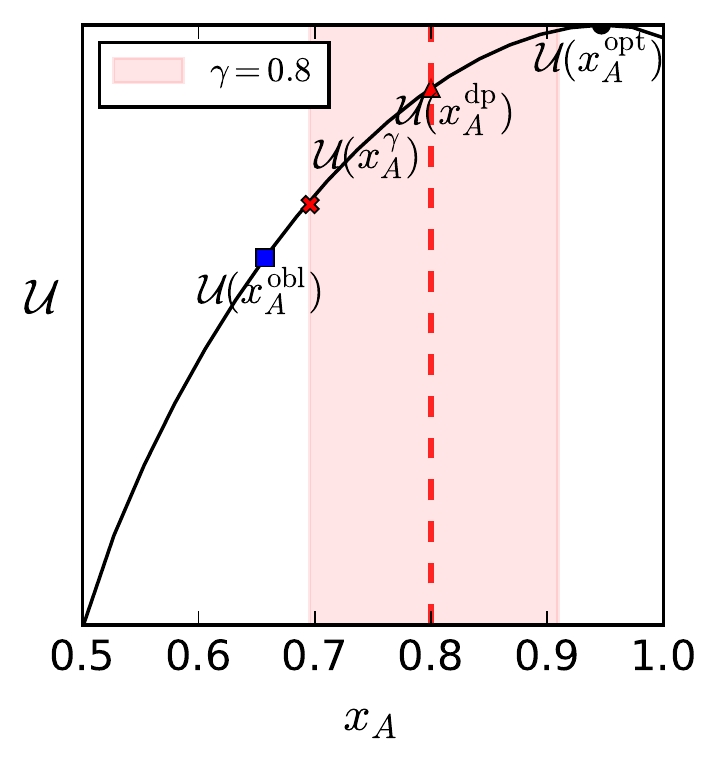} 
    \vspace{-.6cm}\caption{$\alpha_1=0.8$}
    \label{fig: utility 1 stage_d}
  \end{subfigure}
  \vspace{-.2cm}\caption{\textbf{One-stage selection}: Utility $\Q(\pxa)$ as a
    function of selection probability $\pxa$. The function is concave
    and attains its maximum in $\pxaopt$. The parameters are $\mu_W=1$,
    $\sigma_W=1$, $\pa=0.4$, $\sxa=3$, and $\sxb=0.2$.}
  \label{fig: utility 1 stage}
\end{figure}

\subsection{Fairness mechanisms improve selection quality}
\label{ssec:one-stage:improve}

In our work we ask the following question: can fairness mechanisms be beneficial to the utility of a selection process?  A positive answer is given in the following theorem where we show that, for one-stage selection, demographic parity always increases the average quality of a selection compared to the group oblivious algorithm.  Our result also shows that the softer $\gamma$-rule lies between the group oblivious and demographic parity selection algorithms.

\begin{theorem}[Fairness Mechanisms Improves Selection Utility]
  \label{theorem: 1 stage}
  For the one-stage selection problem and for any $\ax\ne1/2$, the
  demographic 
  parity selection algorithm provides a larger utility than a $\gamma$-rule selection algorithm with $\gamma<1$, which in
  turns provides a larger utility than the group-oblivious selection algorithm:
  \begin{align*}
    \Qdp>\Qfair\ge\Qgreedy.
  \end{align*}
  The above inequality is an equality when $\alpha_1=1/2$.
\end{theorem}

\begin{proof}
  Fig.~\ref{fig: utility 1 stage} is a good illustration of this
  proof.  As the selection size is equal to $\alpha_1$, we have
  $\pxa \pa + \pxb \pb = \ax$. By using the results in Lemma
  \ref{lemma: selections for dp and greedy 1 stage} and Lemma
  \ref{lemma: optimal selection}, this implies that:
  \begin{itemize}
  \item When $\ax<1/2$, we have
    $\pxaopt < \pxadp=\ax < \pxafair \le \pxagreedy$. We observe
    that in Fig.~\ref{fig: utility 1 stage_a}b.
  \item When $\ax>1/2$, we have
    $\pxaopt > \pxadp=\ax > \pxafair \ge \pxagreedy$. We observe
    that in Fig.~\ref{fig: utility 1 stage_c}d.
  \end{itemize}
  The results then follow from the concavity of $\Q$ proven earlier. 
%
\end{proof}

Demographic parity helps the selection utility by reducing the effect of implicit variance, but it is also interesting to see how large this performance gap can be.
In Fig.~\ref{fig: all 1 stage}, we show the obtained utilities $\Q$, the selection fractions $\pxa$ and the gap values $(\Qdp{}-\Qgreedy{})/\Qgreedy{}$ for different budgets $\ax$ from 0.01 to 0.99.  Fig.~\ref{fig: all 1 stage_a} illustrates the utilities corresponding to different selection algorithms.  We observe that demographic parity outperforms group oblivious selection, which corresponds to the result of Theorem \ref{theorem: 1 stage}. We also observe that the utilities of the Bayesian-optimal and demographic parity selections decrease with $\alpha_1$. This is expected because this graph represents the average quality of a candidate: the average quality decreases with the number of selected candidates. What is more surprising is that the behavior of the group oblivious selection algorithm is not monotonous: the expected utility $\Q$ increases when $\alpha_1$ goes from $0.1$ to $0.3$. In fact, when $\alpha_1<0.1$, very few $B$-candidates are selected by the group oblivious algorithm. When $\alpha_1>0.1$, this algorithm selects a few good $B$-candidates which leads to an increased average performance.

In Fig.~\ref{fig: all 1 stage_c} we show the performance gap between group oblivious and demographic parity selection algorithms for different values of $\sxa$ and fixed $\sxb=0.2$, $\sq=1$. The values of $\sxa$ are such that $\sxa/\sxb=k$, $k=1,5,10,15$. We see that the gap is in general larger when the selection size $\ax$ is small. This is due to the fact that as the selection size increases, the selections by the group oblivious and demographic parity algorithms become close. The performance gap is zero when $\alpha_1 = 0.5$ because the selections are exactly the same (due to the symmetry of the underlying quality distribution), but it becomes positive again for larger values of $\alpha_1$. In addition, the larger the implicit variance ratio $\sxa^2/ \sxb^2$, the larger the gain that demographic parity brings.

\begin{figure}
    \centering
    \begin{subfigure}{.28\textwidth}
      \includegraphics[width=\linewidth]{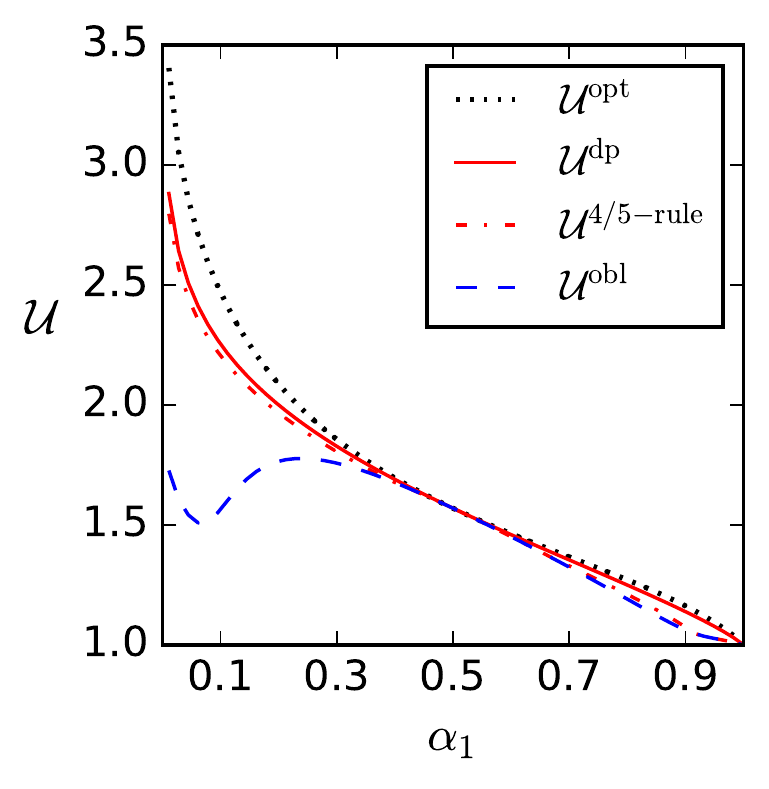}
      \vspace{-.6cm}\caption{Utility}
      \label{fig: all 1 stage_a}
    \end{subfigure}
    \begin{subfigure}{.28\textwidth}
      \includegraphics[width=\linewidth]{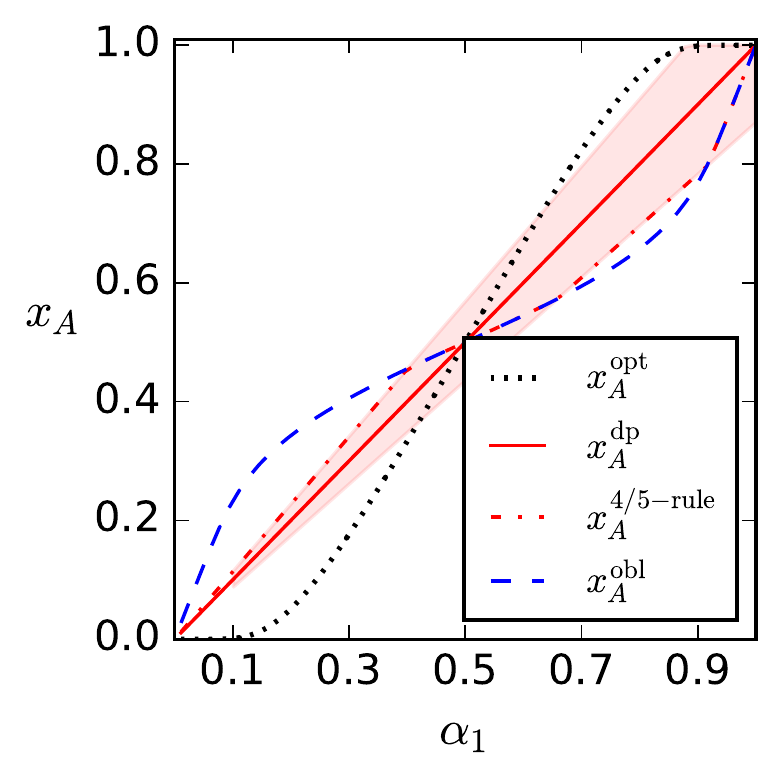}
      \vspace{-.6cm}\caption{Selection fraction}
      \label{fig: all 1 stage_b}
    \end{subfigure}
    \begin{subfigure}{.28\textwidth}
      \includegraphics[width=\linewidth]{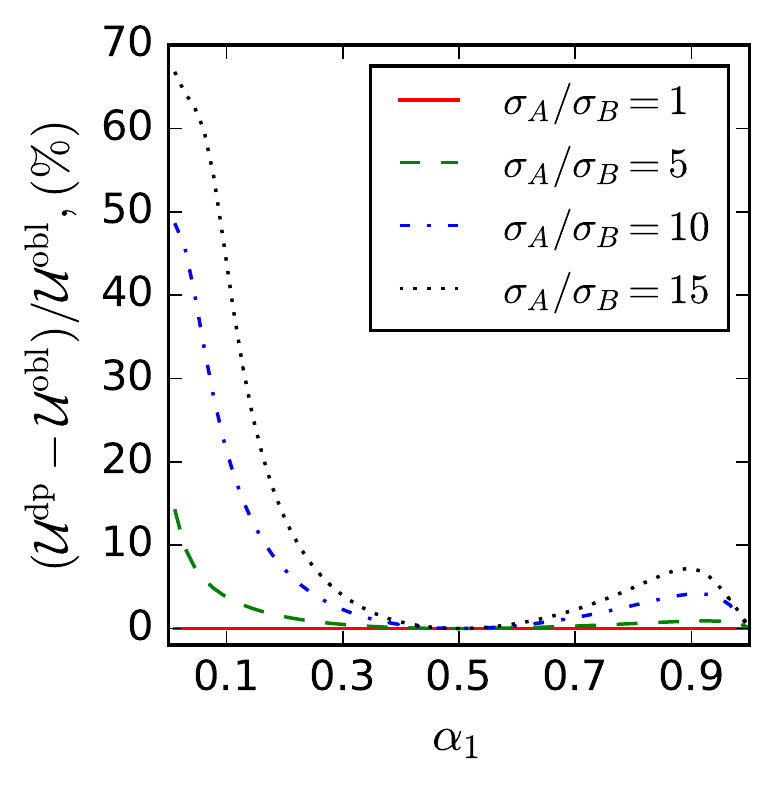}
      \vspace{-.6cm}\caption{Performance gap}
      \label{fig: all 1 stage_c}
    \end{subfigure}
    \vspace{-.2cm}\caption{\textbf{One-stage selection}: Utility $\Q$, selection
      fraction $\pxa$ and performance gap $(\Qdp-\Qgreedy)/ \Qgreedy{}$ for
      different budgets $\ax$. The parameters are $\mq=1$,
      $\sq=1$, $\sxb=0.2$, and $\pa=0.4$; $\sxa=3$ for
      panels~(a,b). \vspace{-2mm}}
    \label{fig: all 1 stage}
\end{figure}




\section{Two-Stage Selection}
\label{section: two stage}

In the previous section, we considered one-stage selection. In reality many processes are in multiple stages, one typical example being hiring. Candidates are first evaluated using tests. After, the examination proceeds with higher accuracy, for instance by performing an interview for every candidate selected at first stage. In our model, first stage results are represented by low accuracy estimates $\wh$ and second stage results are the true values of quality $\w$.  Let us recall that the first-stage fraction of selected candidates is denoted by $\ax$, and the second-stage fraction is $\ay$.

As for the one-stage process, we compare the performance obtained when using at first stage one of the algorithms introduced in Section~\ref{sec:selection_algorithms}. The selection properties of the demographic parity, $\gamma$-rule and group oblivious algorithms are the same as for the one-stage in Section \ref{ssec:one-stage:policies}: group oblivious tends to select a larger fraction of $A$-candidates if the selection size $\ax$ is small and a smaller fraction of $A$-candidates if the selection size is large. Demographic parity and $\gamma$-rule preserve the demographics of candidates during the selection.

For a first-stage selection size $\ax$ close to $\ay$, we expect the two-stage selection to have a behavior similar to one-stage. The behavior for large selection sizes is not obvious, but we expect that both the demographic parity and group oblivious selection algorithms should give similar performance, since there is always a chance to fix the selection at second stage because of the large first stage selection. In fact, the following theorem shows that both for $\alpha_1$ close enough to $\alpha_2$ and for $\alpha_1$ large enough, the demographic parity algorithm leads to a higher average utility than the group-oblivious one. We believe that the case of $\ax$ close enough to $\ay$ is close to what would happen in a number of real selection problems: the proportion of candidates that the decision maker is able to preselect is not too much higher than the total amount of candidates needed.

\begin{theorem}
  \label{theorem: two stage}
  For any problem parameters, there exists $\ax^*$ such that if
  $\ax<\ax^*$ or $\ax>1/2$, then imposing a fairness mechanism at first stage improves the
  utility $\Q$ of the two-stage selection process:
  \begin{align*}
    \Qdp > \Qfair \ge \Qgreedy.
  \end{align*}
    The above inequality is an equality when $\alpha_1=1/2$. 
\end{theorem}

\begin{proof}[Sketch of Proof]
  \emph{Case $\alpha_1<\alpha_1^*$}. We showed in Theorem~\ref{theorem:
    1 stage} that for one-stage selection, demographic parity
  always leads to a better selection. Recall that the case
  $\ax = \ay$ corresponds to one-stage selection. Thus, if
  $\ax = \ay$, then $\Q(\pxadp) > \Q(\pxafair) \ge  \Q(\pxagreedy)$. Since
  $\Q(\pxadp)$, $\Q(\pxafair)$ and $\Q(\pxagreedy)$ are continuous functions of $\ax$ (for a fixed $\ay$),
  we can always find $\ax^*$ (that may depend on $\ay$), such that $\forall \ax < \ax^*$,
  $\Q(\pxadp) > \Q(\pxafair) \geq \Q(\pxagreedy)$ holds.

  \emph{Case $\alpha_1>1/2$}. This second case is harder to study. The key difficulty is to compute the first and second derivatives of $\Q$ with respect to $\pxa$, where $\Q$ is now the expected utility of the final selection (after selecting the candidates with the highest true quality at second stage). The expression for $\frac{d\Q}{d\pxa}$ allows us to check whether a particular algorithm selection fraction $\pxa$ is smaller or larger than the Bayesian-optimal selection fraction $\pxaopt$. We use the expression of the demographic parity selection fraction $\pxadp$ and substitute it into the expression of $\frac{d\Q}{d\pxa}$. We obtain that for $\ax > 1/2$, the derivative of $\Q$ is positive, which means that $\pxadp < \pxaopt$. At the same time, for $\ax > 1/2$, from Lemma~\ref{lemma: selections for dp and greedy 1 stage}, we have $\pxagreedy < \pxadp$. We prove that $\Q$ is strictly concave, again using Harris inequality \cite{fortuin71}, and thus we conclude that $\Qdp > \Qgreedy$.  A detailed proof is provided in Appendix~\ref{proof: two stage}.
\end{proof}







Theorem~\ref{theorem: two stage} provides sufficient conditions under which the $\gamma$-rule and demographic parity algorithms will improve utility. To illustrate this gain, we plot in Fig.~\ref{fig: u2 and p1a normal_a} the utility values obtained when using the Bayesian optimal, demographic parity, $\gamma$-rule or group oblivious selection algorithm. We fix the second stage selection size $\ay=0.1$ and vary the first stage selection size $\ax\in[0.1,1]$.

We observe that for this example, the value of $\ax^*$ in Theorem~\ref{theorem: two stage} seems to be around $0.2$. For $\ax<\ax^*$, the demographic parity or $\ff$-rule algorithms provide a large gain. When $\ax \in (0.2, 0.5)$, the demographic parity and $\gamma$-rule algorithms are not as good as the group-oblivious one but the loss of quality is minimal.  The $\gamma$-rule algorithm provides a compromise between demographic parity and group oblivious: it provides a smaller gain when demographic parity is better than group oblivious; but it provides a smaller loss when demographic parity is not as good as group oblivious.

In Fig.~\ref{fig: u2 and p1a normal_c}, we show the performance gap for different cases of implicit variance values. We observe that in general the larger the implicit variance, the larger the gain that demographic parity brings. The gain can be up to 40\% in our experiments. In general, as the budget $\ax$ grows, we observe a smaller performance gap between the demographic parity and group oblivious algorithms. When demographic parity harms the utility, the harm does not exceed 2\%.
\begin{figure}[ht]
    \centering
\begin{subfigure}{.31\textwidth}
    \includegraphics[width=\linewidth]{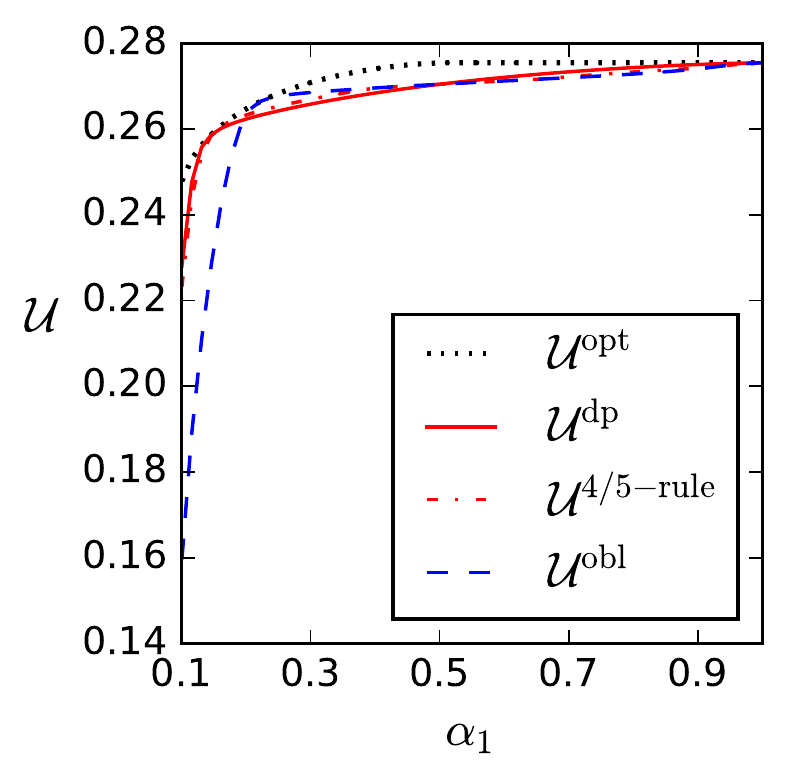}
    \vspace{-.6cm}\caption{Utility}
    \label{fig: u2 and p1a normal_a}
\end{subfigure}
\begin{subfigure}{.3\textwidth}
    \includegraphics[width=\linewidth]{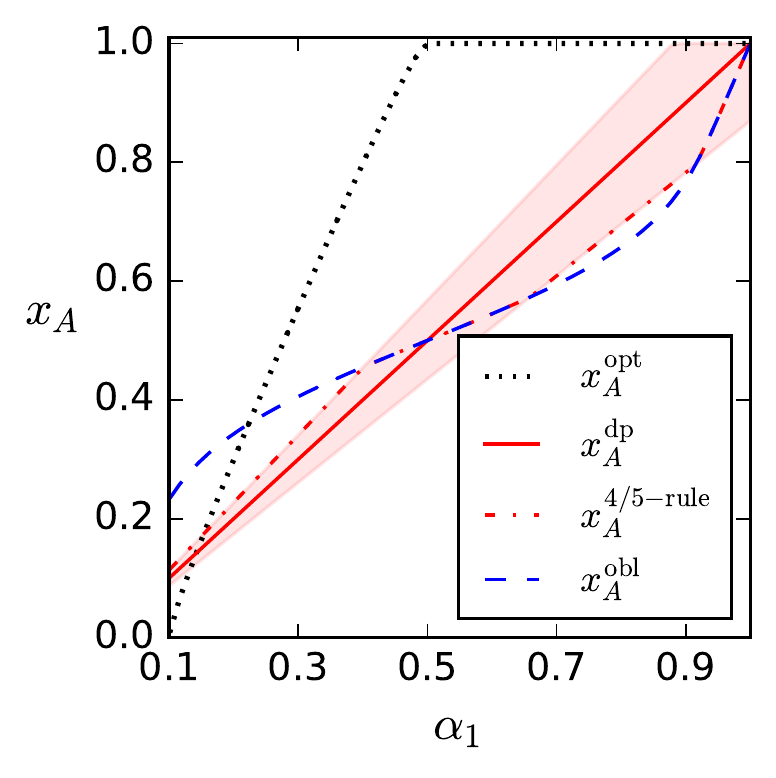}
\vspace{-.6cm}\caption{Selection fraction} 
    \label{fig: u2 and p1a normal_b}
\end{subfigure}
\begin{subfigure}{.3\textwidth}
    \includegraphics[width=\linewidth]{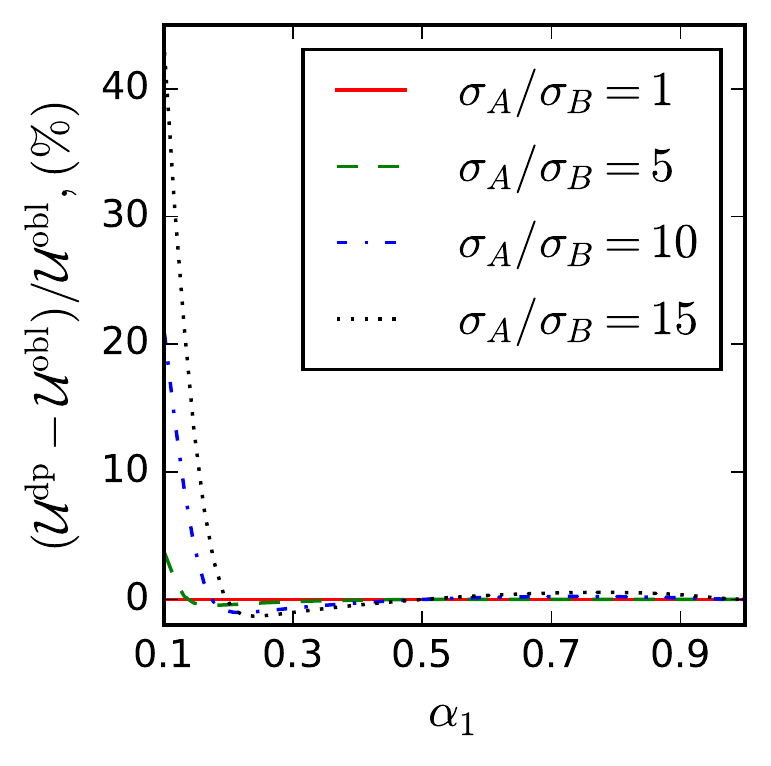}
\vspace{-.6cm}\caption{Performance gap}
    \label{fig: u2 and p1a normal_c}
\end{subfigure}
\vspace{-.2cm}\caption{
Utility of \textbf{two-stage} selection $\Q$, selection fraction
$\pxa$ and performance gap $(\Qdp - \Qgreedy)/ \Qgreedy{}$ for different budgets
$\ax$. The parameters are $\pa=0.4$, $\sxb=0.2$,
$\ay=0.1$; $\sxa$ is set to $\sxa=3$ in (a) and (b). 
}
\label{fig: u2 and p1a normal}
\end{figure}

\section{Experiments}
\label{section: experiments}

In this section,\footnote{All codes are available at:
  \url{https://gitlab.inria.fr/vemelian/implicit-variance-code}
  \href{https://archive.softwareheritage.org/swh:1:rev:196b3b43ff92bf55cd9cd648df1432b339f6c905;origin=https://gitlab.inria.fr/vemelian/implicit-variance-code;visit=swh:1:snp:fc4463ba3f863c6f0382904d5756fd7bbac4c456/}{[permalink
    on softwareheritage.org]}.} we challenge our theoretical results by using sets of data
that do not satisfy our assumptions. We show in
Section~\ref{ssec:synthetic} that the results are qualitatively
similar when the candidates' true quality comes from a Pareto distribution. We
also observe a similar behavior when considering in
Section~\ref{ssec:jee} a real dataset coming from the national Indian
exam data. We conclude in
Section~\ref{ssec:approximation} with experiments that show that a case with $n=20$ candidates
behaves similarly as with $n=\infty$.

\subsection{Synthetic data with Pareto quality}
\label{ssec:synthetic}

Our assumption in the theoretical evaluation of Sections~\ref{section:
  one stage} and \ref{section: two stage} was that qualities $\w$
follow a normal distribution. In some cases, however, the quality
distribution is quite different from normal and can be better modeled
by a power law \cite{kleinberg18}, this for example the case for wealth, income or number of citations \cite{Clauset09a}, meaning that a minority possesses a large fraction of the aggregate quality. In this experiment, we suppose that $\w \sim \Pow(w_0, \kappa)$, where $w_0 > 0$ is a scale and $\kappa > 0$ is a shape parameter: the probability density function of $\w$ can be written as $p_\w(w) = \frac{\kappa w_0^\kappa}{w^{\kappa+1}}$.  We generate $100$ datasets of size $n=10,000$. For every dataset we perform a group oblivious and demographic parity selection. In Fig.~\ref{fig:pareto}, we report the average utilities $\langle \Q_n\rangle$ over the $100$ experiments.

In Fig.~\ref{fig:pareto_a}, we show the performance gap between different
selection algorithms for one-stage selection. We see that demographic
parity improves the utility in most of the cases and that the largest gap
corresponds to the smallest budget $\ax$. Note that contrary to
Theorem~\ref{theorem: 1 stage}, demographic parity does not always
improve utility (for instance here when $\ax\in[.3,.5]$). Yet, the
loss due to demographic parity is never larger than $0.1\%$ while the
gain can be up to $40\%$.  In Fig.~\ref{fig:pareto_b}, the two-stage
case is shown. As expected, demographic parity helps utility for small
budgets $\ax$ close to $\ay=0.01$, since the selection is almost the same as in
one stage. As the budget $\ax$ increases, both the demographic parity and
group oblivious selection algorithms tend to perform close to each other due to
large number of choices at the second stage.  Finally, in
Fig.~\ref{fig:pareto_c}, we show how the selection fraction
$\pxagreedy$ depends on $\ax$. We see that for small budgets $\ax$,
 the group oblivious algorithm tends to select more from group $A$, while for
large budgets, the situation is opposite. 

\begin{figure}[ht]
  \vspace{-.2cm}
  \centering
  \begin{subfigure}{.30\textwidth}
    \includegraphics[width=.98\linewidth]{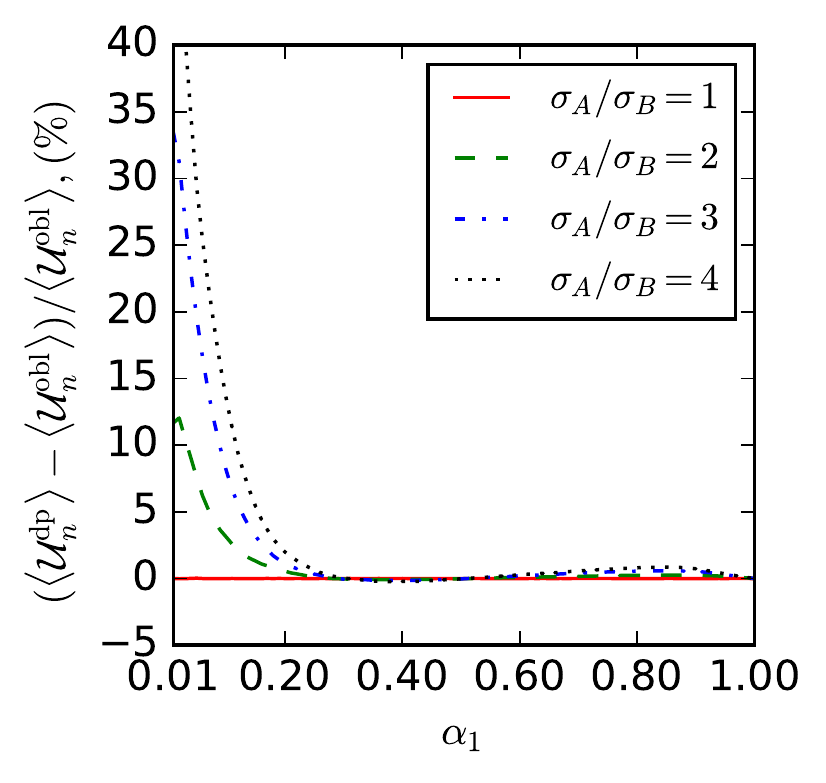}
    \vspace{-.3cm}\caption{One-stage selection}
    \label{fig:pareto_a}
  \end{subfigure}
  \begin{subfigure}{.30\textwidth}
    \includegraphics[width=.98\linewidth]{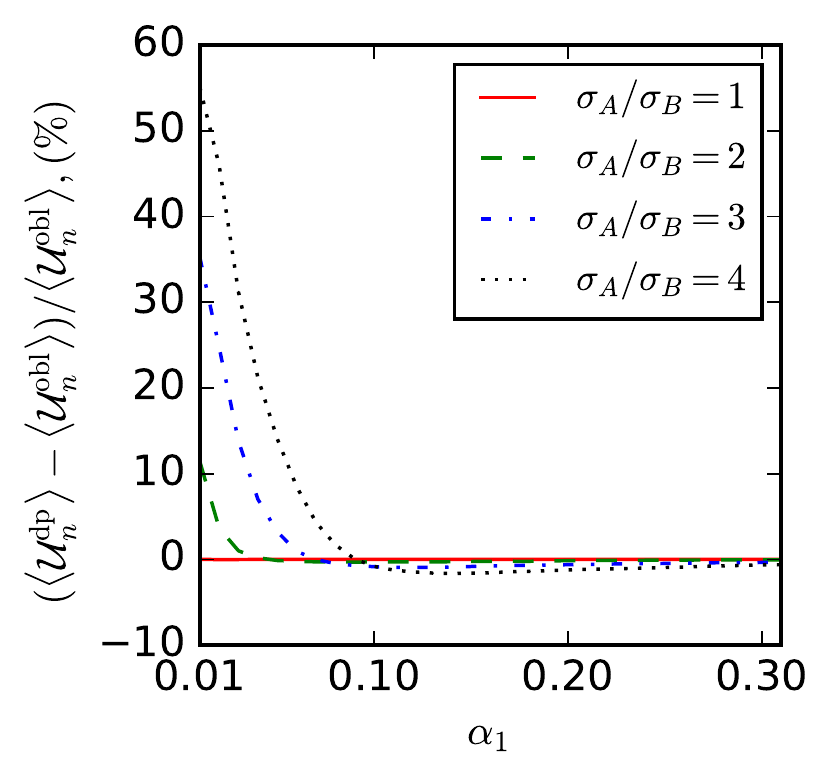}
    \vspace{-.3cm}\caption{Two-stage selection ($\alpha_2{=}0.01$)}
    \label{fig:pareto_b}
  \end{subfigure}
  \begin{subfigure}{.30\textwidth}
    \includegraphics[width=.98\linewidth]{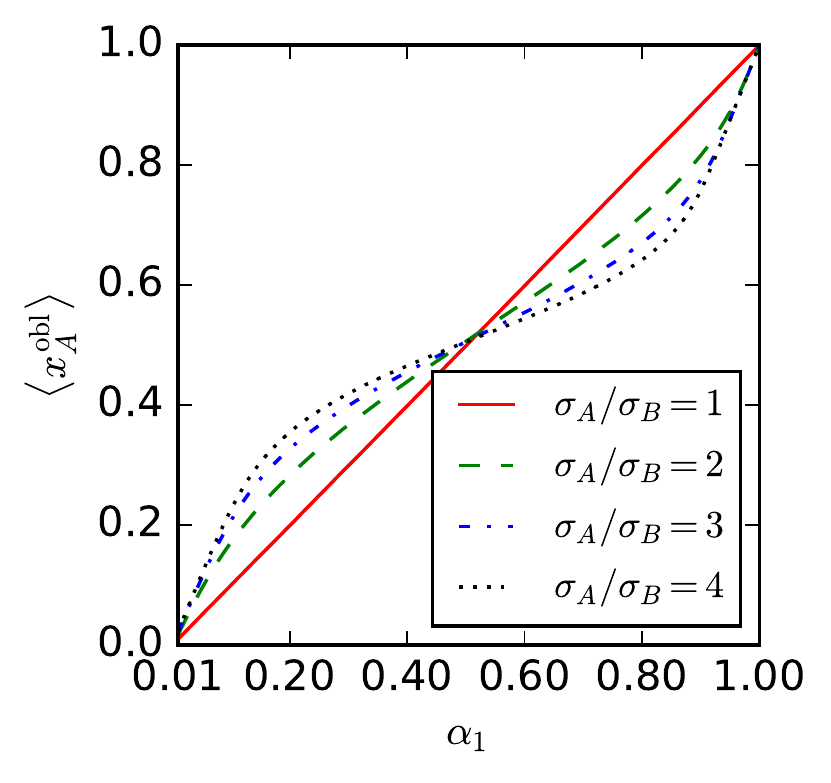} 
    \vspace{-.3cm}\caption{Selection fraction}
    \label{fig:pareto_c}
  \end{subfigure}
  \vspace{-.3cm}\caption{ \textbf{Synthetic data} with Pareto distribution
    $\w \sim \Pow(1,3)$: Gain of demographic parity over group
    oblivious for the one- and two-stage selections. The parameters
    are $\pa=0.4$, $\sxb=1$, $\sxa\in\{1,2,3,4\}$.  }
  \label{fig:pareto}
\end{figure}

Similar results are obtained for other distributions of quality $\w$. We consider uniform, Gaussian mixture and beta distributions. The corresponding plots can be found in Appendix \ref{section: additional figures}.

\subsection{IIT-JEE scores dataset}
\label{ssec:jee}

In this section, we consider a real dataset, the IIT-JEE dataset \cite{jee09}, with joint entrance exam results in India in 2009. These scores are used as an admission criteria to enter the high-rated universities. The dataset consists of 384,977 records. Every record has information about one student: its name, gender, grade for Mathematics, Physics, Chemistry and total grade. In the dataset, there are 98,028 women and 286,942 men.  This dataset is the same as the one considered in \cite{celis20}.

In order to construct a model of implicit variance, we consider an artificial scenario where the field ``grade'' is the true latent quality $\w$ of the candidates. The mean values and standard deviations of $W$ for the two groups are: $\mu_{\w_\mathrm{men}} = 30.8$, $\sigma_{\w_\mathrm{men}} = 51.8$, $\mu_{\w_\mathrm{women}} = 21.2$, $\sigma_{\w_\mathrm{women}} = 39.3$. We then suppose that an unbiased estimator $\wh$ of the grade is observed at the first stage. The standard deviation of estimation for male candidates is set to $\sigma_\mathrm{m} = 10$. For the women group, which is the minority group, we consider different cases: $\sigma_\mathrm{w} = k\cdot\sigma_\mathrm{m}$, for $k=1,4,7,10$.  The distribution of grades $\w$ and observed values $\wh$ for $k=4$ are shown in Fig.~\ref{fig:jee curves_a} and \ref{fig:jee curves_b}.


We start our experiment with \textbf{one-stage} selection. For the dataset we perform a group oblivious (select best $m_1$) and demographic parity selection (select best $m_1$, but maintain the demographic parity condition $\pxa=\pxb$ up to one candidate). The selection size varies from  2\% to 100\% of total number of candidates, i.e., out of 384,977 students the decision maker selects 7,700 students or more. A selection rate of 2\% was set by IIT in 2009 \cite{celis20}.

The results for one-stage selection are given in Fig.~\ref{fig:jee curves_c}. We observe that for both small and large values of $\ax$ demographic parity helps utility, if the noise values of women evaluation $\sigma_\mathrm{w}$ are large. We see that the gain can be up to around 30\% if the selection size is small and up to 5\% if the selection size is large. For the case where $\sigma_\mathrm{w}$ and $\sigma_\mathrm{m}$ are close, we observe no gain if the selection is large and we observe a minor loss in utility (around 2\%) if the selection is small. This is due to the fact that in the dataset, there are more men with a high true latent quality $\w$, as seen in Fig.~\ref{fig:jee curves_a}.

\begin{figure}
  \centering
  \begin{subfigure}{.24\textwidth}
    \includegraphics[width=\linewidth]{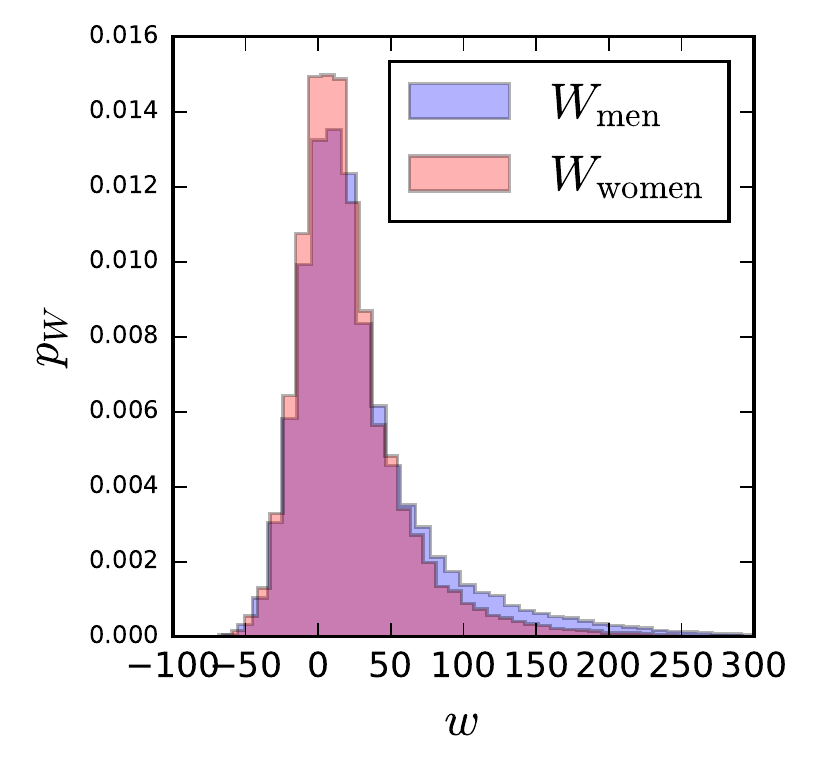}
    \vspace{-.7cm}\caption{Histogram of $\w$}
    \label{fig:jee curves_a}
  \end{subfigure}
  \begin{subfigure}{.24\textwidth}
    \includegraphics[width=\linewidth]{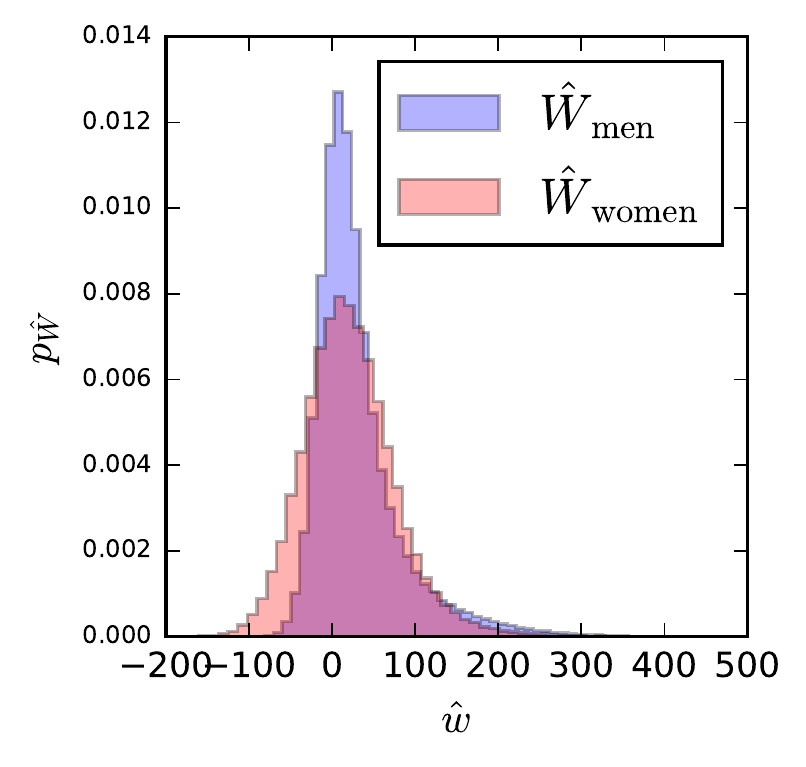}
    \vspace{-.7cm}\caption{Histogram of $\wh$}
    \label{fig:jee curves_b}
  \end{subfigure}
  \begin{subfigure}{.24\textwidth}
    \includegraphics[width=\linewidth]{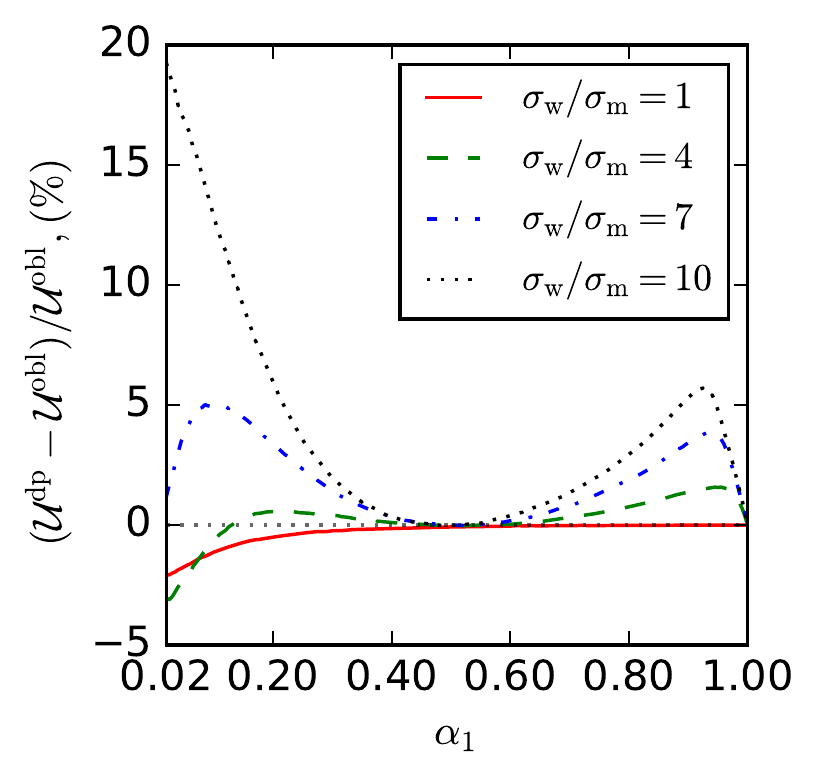} 
    \vspace{-.7cm}\caption{1-stage selection}
    \label{fig:jee curves_c}
  \end{subfigure}
  \begin{subfigure}{.24\textwidth}
    \includegraphics[width=\linewidth]{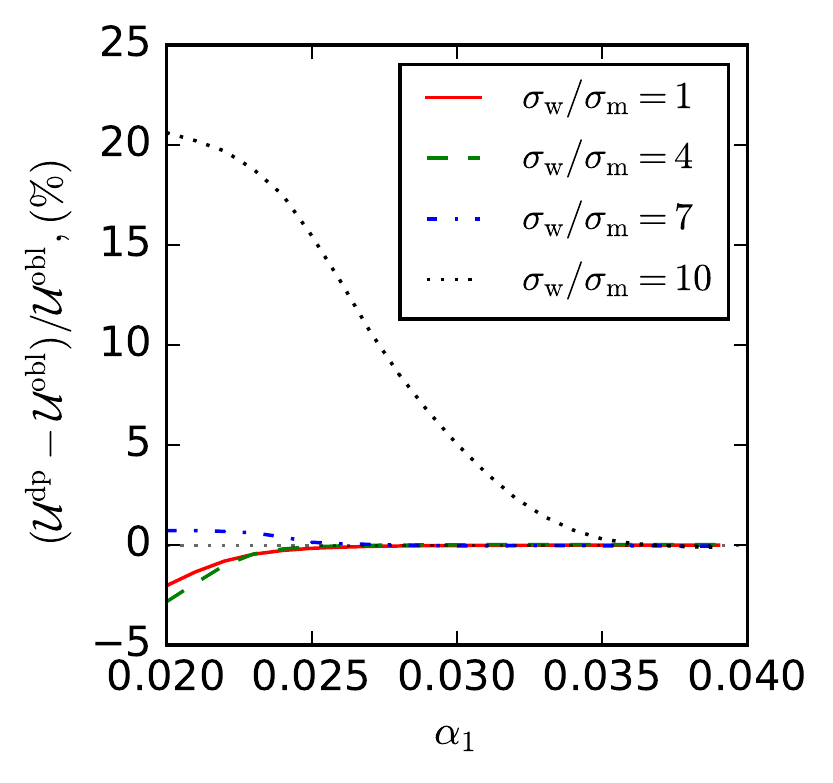}
    \vspace{-.7cm}\caption{2-stage selection}
    \label{fig:jee curves_d}
  \end{subfigure}
  \vspace{-.2cm}\caption{   
    Distribution of $\w$ and $\wh$ given gender, and one- and two-stage selection for \textbf{IIT-JEE dataset} \cite{jee09}. Mean values and standards deviations of $W$ for two groups are: $\mu_{\w_\mathrm{men}} = 30.8$, $\sigma_{\w_\mathrm{men}} = 51.8$,  $\mu_{\w_\mathrm{women}} = 21.2$, $\sigma_{\w_\mathrm{women}} = 39.3$. Added noise has standard deviation $\sigma_\mathrm{m}=10$ and $\sigma_\mathrm{w} = k\cdot\sigma_\mathrm{m}$; $k=4$ in plot (b).}
  \label{fig:jee curves}
  \vspace{-.2cm}
\end{figure}

We now analyze \textbf{two-stage} selection. As for the one-stage case, we perform a group-oblivious selection (select best $m_1$; then select best $m_2$ out of $m_1$) and a demographic parity selection (select best $m_1$ but maintain  demographic parity; then select best $m_2$ out of $m_1$). In Fig.~\ref{fig:jee curves_d} we show the case where the final-stage selection rate is $\ay=2\%$, i.e., we select 7,700 candidates out of 384,977. As observed, the performance gain when using demographic parity can be up to 20\%. However, as the selection size $\ax$ increases, $\Qdp$ and $\Qgreedy$ become close, since there will always be enough candidates among those selected at first stage to subselect a tiny proportion of good candidates.


\subsection{Accuracy of the approximation for small $n$}
\label{ssec:approximation}

As discussed in Section~\ref{section: model}, we cannot solve the problem with finite selection sizes exactly. Instead, we use an approximation that is exact as number of candidates $n$ tends to infinity (Proposition~\ref{th:n_infinity}). However, it is important to know how the approximation behaves for a small number of candidates $n$ and small selection sizes $m_1$, $m_2$. For our experiment, we generate datasets of different sizes $n=20, 50, 100$. For every size parameter $n$, we generate $10,000$ different datasets. For a population size $n$, we denote by $\langle\Q_n\rangle$ the average quality of the selected candidates over our $10,000$ experiments. In each case, the true latent qualities $\w$ are generated from a normal distribution $\N(1,1)$.


In Fig.~\ref{fig: finite normal_a} we plot the average utilities
$\langle\Q_n\rangle$ for a population of $n=100$, where we select
$m_2=10$ individuals and where we vary $m_1$ from 10 to 100. The
shaded region corresponds to a confidence interval.  We consider two
selection algorithms (demographic parity and group oblivious) and
compare the performance for $n=100$ with the limiting quantities
$\Qdp$ and $\Qgreedy$.  We observe that, even for $n=100$, the average
values of utility are close to the approximation.  In Fig.~\ref{fig:
  finite normal_b} we compare the gap of average performances $(\langle
\Qdp_n\rangle - \langle\Qgreedy_n\rangle)/\langle\Qgreedy_n\rangle$ for different $n$. We observe
that the approximation $n=+\infty$ is a good prediction of the average
gain provided by the use of demographic parity.  In order to
distinguish more precisely over the various experiments, in
Fig.~\ref{fig: finite normal_c} we compare the average gain of
performance $\langle (\Qdp_n - \Qgreedy_n)/ \Qgreedy_n\rangle$. Again, the curves
for finite $n$ are almost indistinguishable from the case where $n\to\infty$. 


\begin{figure}
  \centering
  \begin{subfigure}{.32\textwidth}
    \includegraphics[height=.95\linewidth]{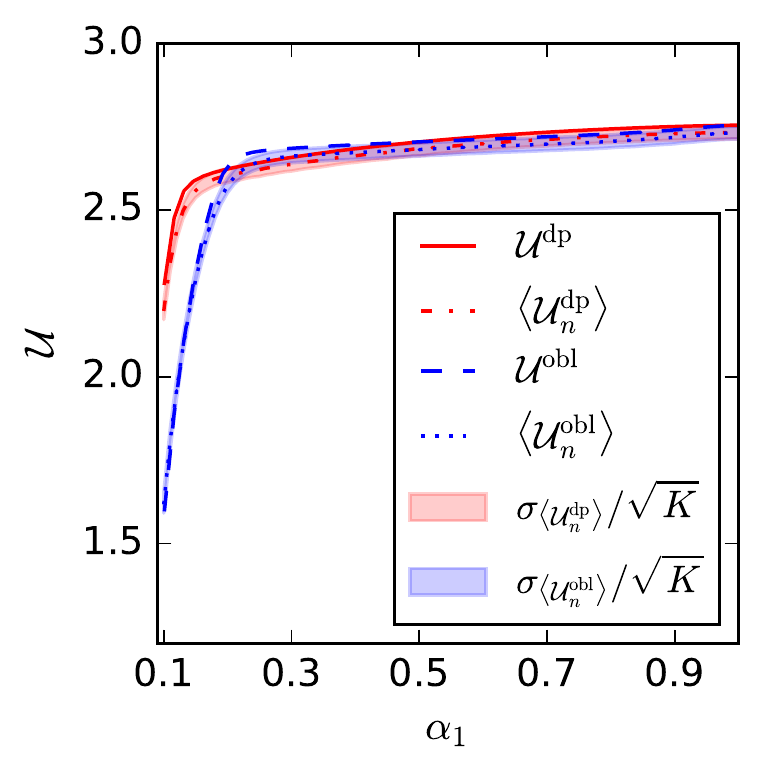} 
    \vspace{-.3cm}\caption{Average utility $\langle \Q_n\rangle$, $n=100$}
    \label{fig: finite normal_a}
  \end{subfigure}
  \begin{subfigure}{.32\textwidth}
    \includegraphics[height=.95\linewidth]{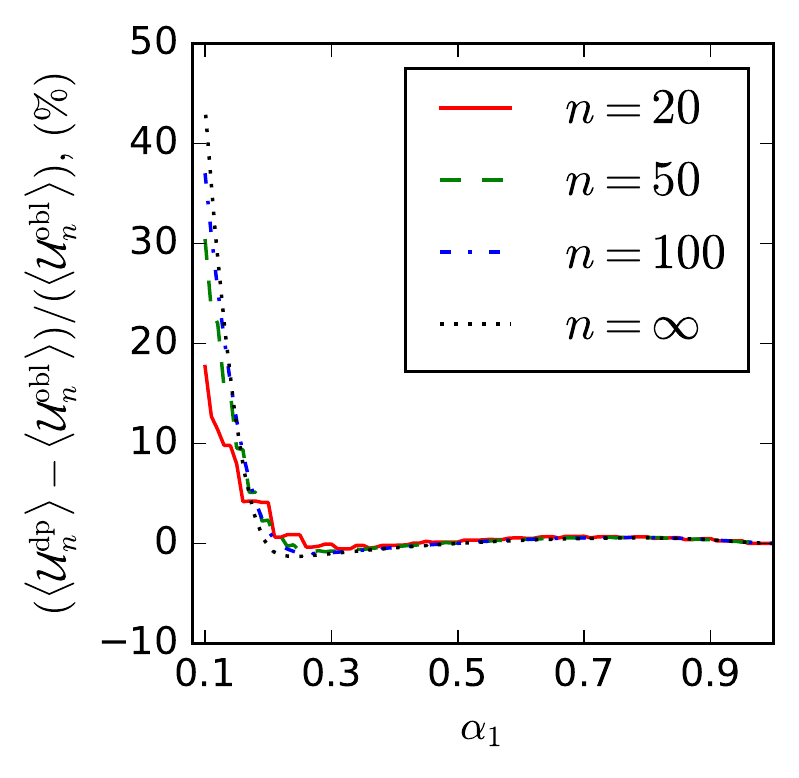}
    \vspace{-.3cm}\caption{Gap of average performance}
    \label{fig: finite normal_b}
  \end{subfigure}
  \begin{subfigure}{.32\textwidth}
    \includegraphics[height=.95\linewidth]{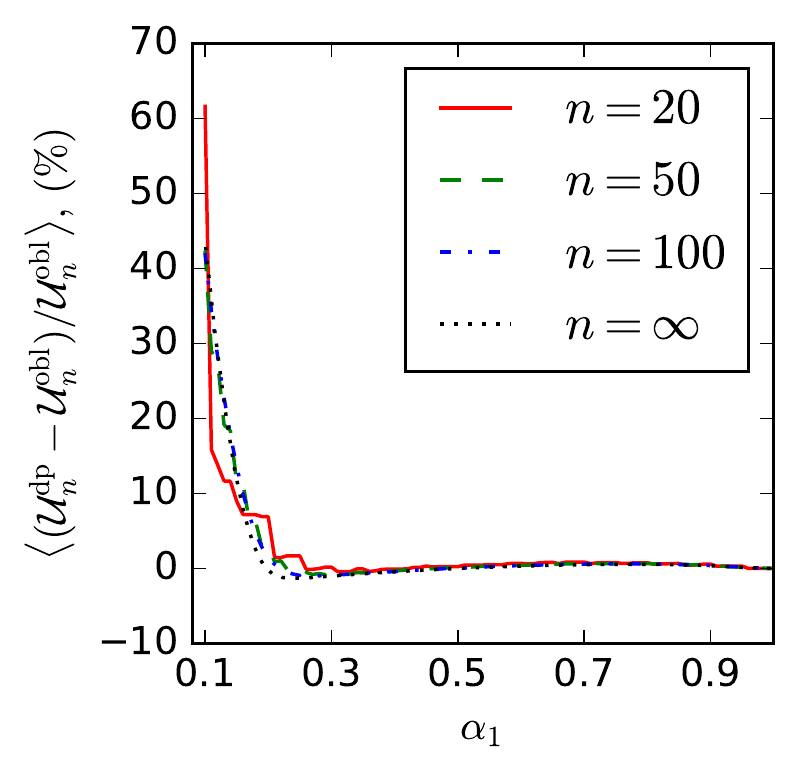}
    \vspace{-.3cm}\caption{Average performance gap}
    \label{fig: finite normal_c}
  \end{subfigure}
  \vspace{-.2cm}\caption{\textbf{Finite population size}: quality of the two-stage selection and expected gain of demographic parity over group oblivious. The quality distribution $\w$ is $\N(1,1)$ and the noise parameters are $\sxa = 3$, $\sxb=0.2$, $\ay=0.1$. The number of experiments per set of parameters is $K=10,000$. The shaded areas are the confidence intervals (corresponding to one standard deviation on the estimation of the empirical mean).}
  \label{fig: finite normal}
\end{figure}

\section{Discussion and Extensions}
\label{section: discussion}

In this work, we study one- and two-stage selection problems in the presence of implicit variance. We propose a purposely simple model of the problem that captures the phenomenon of implicit variance and allows us to obtain clean mathematical results. In particular, we show that fairness mechanisms (a generalization of the $\ff$ rule) often lead to a higher selection utility compared to using a group oblivious algorithm. Our model is flexible and can be extended in several directions.

\paragraph{Different prior of the quality distribution}

Our theoretical results are obtained under the assumption that the true latent quality $\w$ follows a group-independent distribution (to isolate the effect of implicit variance) that corresponds to a normal law (to allow for analytical derivations). Both assumptions can be relaxed. \emph{First}, we can plug into the model any distribution of latent quality (e.g., Pareto, uniform, mixture of Gaussians, etc.). We show numerically in Section~\ref{section: experiments} and Appendix~\ref{section: additional figures} that it does not change the flavor of the main results. \emph{Second}, we can consider quality distributions dependent on the group. A natural extension in that direction would be to consider two different normal distributions. It is possible to extend our results to that case (at the expense of increased complexity). Our experiment on the ITT-JEE dataset (Section~\ref{ssec:jee}), however, gives a preview of how the results are modified: if the effect of implicit variance is small compared to the difference in the true quality distributions then demographic parity can decrease the selection quality for small selection budgets. If the effect of implicit variance is predominant then our results continue to hold.

\paragraph{Combining implicit variance and implicit bias}
Our model does not include implicit bias so as to better isolate the effect of implicit variance. It is easy, though, to incorporate implicit bias as in \cite{kleinberg18,celis20}. The most natural in our model would be to consider $\whi = \wi - \beta_{G_i} + \sigma_{G_i} \varepsilon_i$, where $\beta_{G}$ is the implicit additive bias against group $G$ (typically, $\beta>0$ for the disadvantaged group and $\beta=0$ for the other). In effect, the additive parameter $\beta$ shifts the distribution of $\whi$, while the additive noise widens it. If true qualities are normally distributed, we would then have $\wh_A \sim \N(\mq - \beta_A, \sq^2 + \sxa^2)$ and $\wh_B \sim \N(\mq - \beta_B, \sq^2 + \sxb^2)$. We leave as future work a detailed investigation of the group oblivious and fair selection utilities in that case.

\paragraph{Effect on global fairness in two-stage selection}

Throughout the paper, we have studied the effect of imposing demographic parity at the first stage on the final selection \emph{utility}. However, a natural question to ask is what is the effect of imposing fairness at the first stage on the fairness (or disparity) of the final selection (we term it global fairness following~\cite{emelianov19}). In Fig.~\ref{fig: fairness 2nd stage}, we plot the global selection ratios $\pya$ and $\pyb$ for each group (at the second stage), for the different first-stage algorithms considered in the paper. For selection budgets $\ax$ close to $\ay$, we observe that the first-stage demographic parity algorithm leads to the smallest disparity in global selection fractions. It is natural since the selection fractions are close to ones obtained in the one-stage case (so for demographic parity $\pya \approx \pyb$ since $\pxa = \pxb$). However, we observe that as $\ax$ grows, demographic parity can lead to a larger inequality in the global selection fractions than with the group oblivious and optimal algorithms. Thus, in a two-stage selection problem with implicit variance, imposing a fairness constraint at the first stage may lead to a degradation of the global fairness. We leave a detailed investigation of this counter-intuitive aspect as future work but note that this emphasizes the crucial importance of modeling implicit variance in multistage selection problems.

\begin{figure}
  \includegraphics[width=0.43\linewidth]{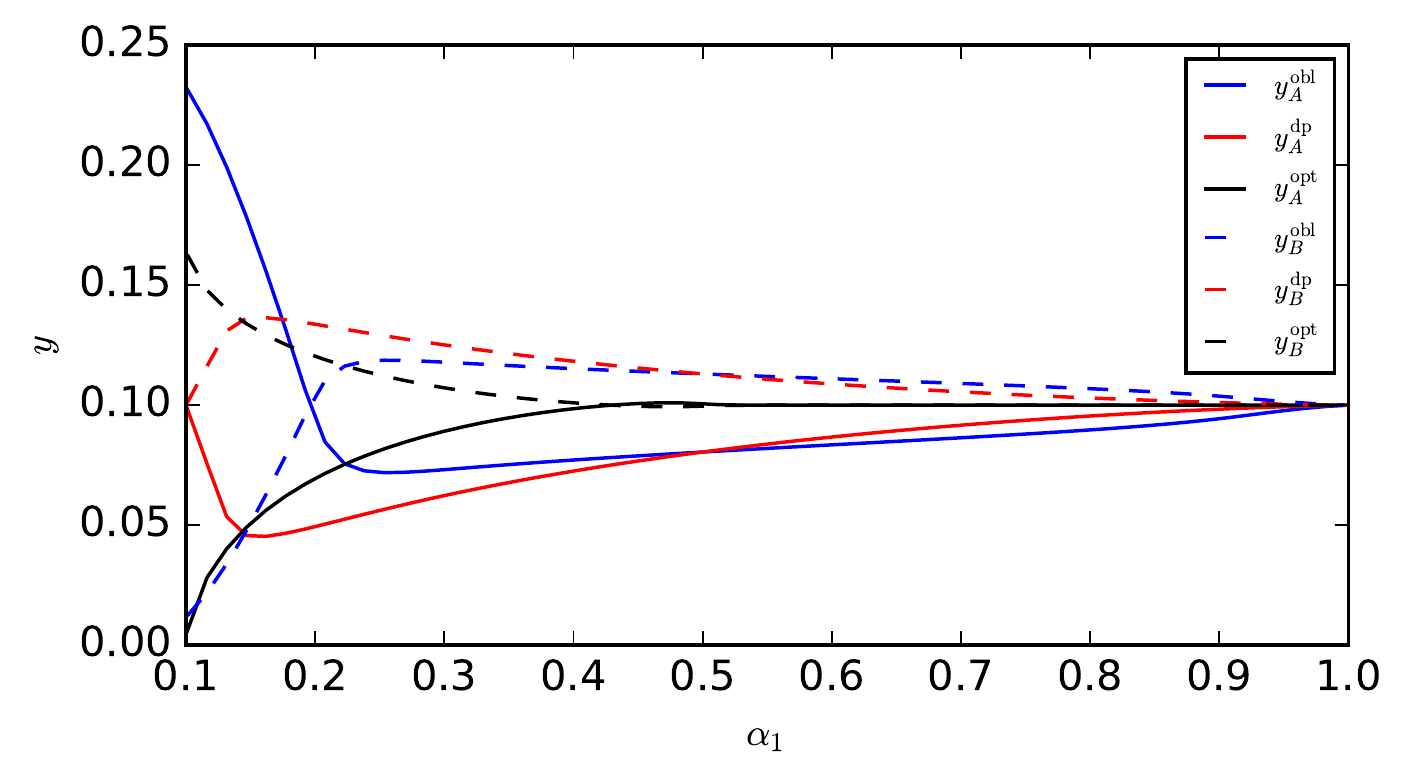}
  \vspace{-.5cm}
  \caption{Selection fractions $y_A$ and $y_B$ corresponding to selection at second stage. Parameters: $\w \sim \N(1,1)$, $\sxa=3$, $\sxb=0.2$, $\pa=0.4$ and $\ay=0.1$} 
  \label{fig: fairness 2nd stage}
  \vspace{-.5cm}
\end{figure}

\begin{acks}
This work has been partially supported by MIAI @ Grenoble Alpes (ANR-19-P3IA-0003) and by a European Research Council (ERC) Advanced Grant for the project ``Foundations for Fair Social Computing'' funded under the European Union's Horizon 2020 Framework Programme (grant agreement no. 789373). We thank the reviewers for their thoughtful comments. 
\end{acks}

\bibliographystyle{ACM-Reference-Format}
\bibliography{bibliography} 


\begin{thebibliography}{30}


\ifx \showCODEN    \undefined \def \showCODEN     #1{\unskip}     \fi
\ifx \showDOI      \undefined \def \showDOI       #1{#1}\fi
\ifx \showISBNx    \undefined \def \showISBNx     #1{\unskip}     \fi
\ifx \showISBNxiii \undefined \def \showISBNxiii  #1{\unskip}     \fi
\ifx \showISSN     \undefined \def \showISSN      #1{\unskip}     \fi
\ifx \showLCCN     \undefined \def \showLCCN      #1{\unskip}     \fi
\ifx \shownote     \undefined \def \shownote      #1{#1}          \fi
\ifx \showarticletitle \undefined \def \showarticletitle #1{#1}   \fi
\ifx \showURL      \undefined \def \showURL       {\relax}        \fi
\providecommand\bibfield[2]{#2}
\providecommand\bibinfo[2]{#2}
\providecommand\natexlab[1]{#1}
\providecommand\showeprint[2][]{arXiv:#2}

\bibitem[\protect\citeauthoryear{??}{jee}{2019}]%
        {jee09}
 \bibinfo{year}{2019}\natexlab{}.
\newblock \bibinfo{title}{{IIT-JEE dataset}}.
\newblock
  \bibinfo{howpublished}{\url{https://github.com/AnayMehrotra/Ranking-with-Implicit-Bias}}.
\newblock
\newblock
\shownote{[Online; accessed Jan 29, 2020].}


\bibitem[\protect\citeauthoryear{Aigner and Cain}{Aigner and Cain}{1977}]%
        {Aigner77a}
\bibfield{author}{\bibinfo{person}{Dennis~J. Aigner} {and}
  \bibinfo{person}{Glen~G. Cain}.} \bibinfo{year}{1977}\natexlab{}.
\newblock \showarticletitle{Statistical Theories of Discrimination in Labor
  Markets}.
\newblock \bibinfo{journal}{\emph{Industrial and Labor Relations Review}}
  \bibinfo{volume}{30}, \bibinfo{number}{2} (\bibinfo{year}{1977}),
  \bibinfo{pages}{175--187}.
\newblock


\bibitem[\protect\citeauthoryear{Balafoutas and Sutter}{Balafoutas and
  Sutter}{2012}]%
        {aff_action_balafoutas12}
\bibfield{author}{\bibinfo{person}{Loukas Balafoutas} {and}
  \bibinfo{person}{Matthias Sutter}.} \bibinfo{year}{2012}\natexlab{}.
\newblock \showarticletitle{Affirmative Action Policies Promote Women and Do
  Not Harm Efficiency in the Laboratory}.
\newblock \bibinfo{journal}{\emph{Science}}  \bibinfo{volume}{335}
  (\bibinfo{date}{Feb.} \bibinfo{year}{2012}), \bibinfo{pages}{579--82}.
\newblock


\bibitem[\protect\citeauthoryear{Baye and Monseur}{Baye and Monseur}{2016}]%
        {gender_variability_baye16}
\bibfield{author}{\bibinfo{person}{Ariane Baye} {and}
  \bibinfo{person}{Christian Monseur}.} \bibinfo{year}{2016}\natexlab{}.
\newblock \showarticletitle{Gender differences in variability and extreme
  scores in an international context}.
\newblock \bibinfo{journal}{\emph{Large-scale Assessments in Education}}
  \bibinfo{volume}{4} (\bibinfo{date}{Dec.} \bibinfo{year}{2016}).
\newblock


\bibitem[\protect\citeauthoryear{Bertrand and Mullainathan}{Bertrand and
  Mullainathan}{2004}]%
        {Bertrand04a}
\bibfield{author}{\bibinfo{person}{Marianne Bertrand} {and}
  \bibinfo{person}{Sendhil Mullainathan}.} \bibinfo{year}{2004}\natexlab{}.
\newblock \showarticletitle{{Are Emily and Greg More Employable Than Lakisha
  and Jamal? A Field Experiment on Labor Market Discrimination}}.
\newblock \bibinfo{journal}{\emph{American Economic Review}}
  \bibinfo{volume}{94}, \bibinfo{number}{4} (\bibinfo{date}{Sept.}
  \bibinfo{year}{2004}), \bibinfo{pages}{991--1013}.
\newblock


\bibitem[\protect\citeauthoryear{Cavicchia}{Cavicchia}{2015}]%
        {Cavicchia15a}
\bibfield{author}{\bibinfo{person}{Marilyn Cavicchia}.}
  \bibinfo{year}{2015}\natexlab{}.
\newblock \showarticletitle{How to fight implicit bias? With conscious thought,
  diversity expert tells NABE}.
\newblock \bibinfo{journal}{\emph{American Bar Association: Bar Leader}}
  \bibinfo{volume}{40}, \bibinfo{number}{1} (\bibinfo{year}{2015}).
\newblock


\bibitem[\protect\citeauthoryear{Celis, Mehrotra, and Vishnoi}{Celis
  et~al\mbox{.}}{2020}]%
        {celis20}
\bibfield{author}{\bibinfo{person}{L.~Elisa Celis}, \bibinfo{person}{Anay
  Mehrotra}, {and} \bibinfo{person}{Nisheeth~K. Vishnoi}.}
  \bibinfo{year}{2020}\natexlab{}.
\newblock \showarticletitle{Interventions for Ranking in the Presence of
  Implicit Bias}. In \bibinfo{booktitle}{\emph{Proceedings of the 2020
  Conference on Fairness, Accountability, and Transparency (FAT*)}}.
  \bibinfo{pages}{369--380}.
\newblock


\bibitem[\protect\citeauthoryear{Chouldechova}{Chouldechova}{2017}]%
        {Chouldechova17a}
\bibfield{author}{\bibinfo{person}{Alexandra Chouldechova}.}
  \bibinfo{year}{2017}\natexlab{}.
\newblock \showarticletitle{Fair Prediction with Disparate Impact: A Study of
  Bias in Recidivism Prediction Instruments}.
\newblock \bibinfo{journal}{\emph{Big Data}} \bibinfo{volume}{5},
  \bibinfo{number}{2} (\bibinfo{year}{2017}), \bibinfo{pages}{153--163}.
\newblock


\bibitem[\protect\citeauthoryear{Clauset, Shalizi, and Newman}{Clauset
  et~al\mbox{.}}{2009}]%
        {Clauset09a}
\bibfield{author}{\bibinfo{person}{Aaron Clauset},
  \bibinfo{person}{Cosma~Rohilla Shalizi}, {and} \bibinfo{person}{M.~E.~J.
  Newman}.} \bibinfo{year}{2009}\natexlab{}.
\newblock \showarticletitle{Power-Law Distributions in Empirical Data}.
\newblock \bibinfo{journal}{\emph{SIAM Rev.}} \bibinfo{volume}{51},
  \bibinfo{number}{4} (\bibinfo{year}{2009}), \bibinfo{pages}{661--703}.
\newblock


\bibitem[\protect\citeauthoryear{Coate and Loury}{Coate and Loury}{1993}]%
        {coate93}
\bibfield{author}{\bibinfo{person}{Stephen Coate} {and} \bibinfo{person}{Glenn
  Loury}.} \bibinfo{year}{1993}\natexlab{}.
\newblock \showarticletitle{Will Affirmative-Action Policies Eliminate Negative
  Stereotypes?}
\newblock \bibinfo{journal}{\emph{American Economic Review}}
  \bibinfo{volume}{83} (\bibinfo{date}{Feb.} \bibinfo{year}{1993}),
  \bibinfo{pages}{1220--40}.
\newblock


\bibitem[\protect\citeauthoryear{Collins}{Collins}{2007}]%
        {rooney_rule_collins07}
\bibfield{author}{\bibinfo{person}{Brian Collins}.}
  \bibinfo{year}{2007}\natexlab{}.
\newblock \showarticletitle{Tackling Unconscious Bias in Hiring Practices: The
  Plight of the Rooney Rule}.
\newblock \bibinfo{journal}{\emph{NYU Law Review}}  \bibinfo{volume}{82}
  (\bibinfo{date}{June} \bibinfo{year}{2007}).
\newblock


\bibitem[\protect\citeauthoryear{Corbett-Davies, Pierson, Feller, Goel, and
  Huq}{Corbett-Davies et~al\mbox{.}}{2017}]%
        {Corbett-Davies:2017}
\bibfield{author}{\bibinfo{person}{Sam Corbett-Davies}, \bibinfo{person}{Emma
  Pierson}, \bibinfo{person}{Avi Feller}, \bibinfo{person}{Sharad Goel}, {and}
  \bibinfo{person}{Aziz Huq}.} \bibinfo{year}{2017}\natexlab{}.
\newblock \showarticletitle{Algorithmic Decision Making and the Cost of
  Fairness}. In \bibinfo{booktitle}{\emph{Proceedings of the 23rd ACM SIGKDD
  International Conference on Knowledge Discovery and Data Mining (KDD)}}.
  \bibinfo{pages}{797--806}.
\newblock


\bibitem[\protect\citeauthoryear{Emelianov, Arvanitakis, Gast, Gummadi, and
  Loiseau}{Emelianov et~al\mbox{.}}{2019}]%
        {emelianov19}
\bibfield{author}{\bibinfo{person}{Vitalii Emelianov}, \bibinfo{person}{George
  Arvanitakis}, \bibinfo{person}{Nicolas Gast}, \bibinfo{person}{Krishna
  Gummadi}, {and} \bibinfo{person}{Patrick Loiseau}.}
  \bibinfo{year}{2019}\natexlab{}.
\newblock \showarticletitle{The Price of Local Fairness in Multistage
  Selection}. In \bibinfo{booktitle}{\emph{Proceedings of the 28th
  International Joint Conference on Artificial Intelligence (IJCAI)}}.
\newblock


\bibitem[\protect\citeauthoryear{Fortuin, Kasteleyn, and Ginibre}{Fortuin
  et~al\mbox{.}}{1971}]%
        {fortuin71}
\bibfield{author}{\bibinfo{person}{C.~M. Fortuin}, \bibinfo{person}{P.~W.
  Kasteleyn}, {and} \bibinfo{person}{J. Ginibre}.}
  \bibinfo{year}{1971}\natexlab{}.
\newblock \showarticletitle{Correlation inequalities on some partially ordered
  sets}.
\newblock \bibinfo{journal}{\emph{Comm. Math. Phys.}} \bibinfo{volume}{22},
  \bibinfo{number}{2} (\bibinfo{year}{1971}), \bibinfo{pages}{89--103}.
\newblock


\bibitem[\protect\citeauthoryear{Greenwald and Krieger}{Greenwald and
  Krieger}{2006}]%
        {implicit_bias_greenwald06}
\bibfield{author}{\bibinfo{person}{Anthony Greenwald} {and}
  \bibinfo{person}{Linda Krieger}.} \bibinfo{year}{2006}\natexlab{}.
\newblock \showarticletitle{Implicit Bias: Scientific Foundations}.
\newblock \bibinfo{journal}{\emph{California Law Review}}  \bibinfo{volume}{94}
  (\bibinfo{date}{July} \bibinfo{year}{2006}), \bibinfo{pages}{945}.
\newblock


\bibitem[\protect\citeauthoryear{Hardt, Price, and Srebro}{Hardt
  et~al\mbox{.}}{2016}]%
        {Hardt:2016}
\bibfield{author}{\bibinfo{person}{Moritz Hardt}, \bibinfo{person}{Eric Price},
  {and} \bibinfo{person}{Nathan Srebro}.} \bibinfo{year}{2016}\natexlab{}.
\newblock \showarticletitle{Equality of Opportunity in Supervised Learning}. In
  \bibinfo{booktitle}{\emph{Proceedings of the 30th International Conference on
  Neural Information Processing Systems (NIPS)}}. \bibinfo{pages}{3323--3331}.
\newblock


\bibitem[\protect\citeauthoryear{Holzer and Neumark}{Holzer and
  Neumark}{2000}]%
        {holzer00}
\bibfield{author}{\bibinfo{person}{Harry Holzer} {and} \bibinfo{person}{David
  Neumark}.} \bibinfo{year}{2000}\natexlab{}.
\newblock \showarticletitle{Assessing Affirmative Action}.
\newblock \bibinfo{journal}{\emph{Journal of Economic Literature}}
  \bibinfo{volume}{38}, \bibinfo{number}{3} (\bibinfo{date}{Sept.}
  \bibinfo{year}{2000}), \bibinfo{pages}{483--568}.
\newblock


\bibitem[\protect\citeauthoryear{Kannan, Roth, and Ziani}{Kannan
  et~al\mbox{.}}{2019}]%
        {kannan19}
\bibfield{author}{\bibinfo{person}{Sampath Kannan}, \bibinfo{person}{Aaron
  Roth}, {and} \bibinfo{person}{Juba Ziani}.} \bibinfo{year}{2019}\natexlab{}.
\newblock \showarticletitle{Downstream Effects of Affirmative Action}. In
  \bibinfo{booktitle}{\emph{Proceedings of the Conference on Fairness,
  Accountability, and Transparency (FAT*)}}.
\newblock


\bibitem[\protect\citeauthoryear{Kleinberg and Raghavan}{Kleinberg and
  Raghavan}{2018}]%
        {kleinberg18}
\bibfield{author}{\bibinfo{person}{Jon~M. Kleinberg} {and}
  \bibinfo{person}{Manish Raghavan}.} \bibinfo{year}{2018}\natexlab{}.
\newblock \showarticletitle{Selection Problems in the Presence of Implicit
  Bias}. In \bibinfo{booktitle}{\emph{Proceedings of the 9th Innovations in
  Theoretical Computer Science Conference (ITCS)}}.
  \bibinfo{pages}{33:1--33:17}.
\newblock


\bibitem[\protect\citeauthoryear{Lipton, McAuley, and Chouldechova}{Lipton
  et~al\mbox{.}}{2018}]%
        {Lipton18a}
\bibfield{author}{\bibinfo{person}{Zachary Lipton}, \bibinfo{person}{Julian
  McAuley}, {and} \bibinfo{person}{Alexandra Chouldechova}.}
  \bibinfo{year}{2018}\natexlab{}.
\newblock \showarticletitle{Does mitigating ML's impact disparity require
  treatment disparity?}. In \bibinfo{booktitle}{\emph{Proceedings of the 32nd
  International Conference on Neural Information Processing Systems (NIPS)}}.
  \bibinfo{pages}{8125--8135}.
\newblock


\bibitem[\protect\citeauthoryear{Locatello, Abbati, Rainforth, Bauer,
  Sch\"{o}lkopf, and Bachem}{Locatello et~al\mbox{.}}{2019}]%
        {locatello19}
\bibfield{author}{\bibinfo{person}{Francesco Locatello},
  \bibinfo{person}{Gabriele Abbati}, \bibinfo{person}{Thomas Rainforth},
  \bibinfo{person}{Stefan Bauer}, \bibinfo{person}{Bernhard Sch\"{o}lkopf},
  {and} \bibinfo{person}{Olivier Bachem}.} \bibinfo{year}{2019}\natexlab{}.
\newblock \showarticletitle{On the Fairness of Disentangled Representations}.
\newblock In \bibinfo{booktitle}{\emph{Proceedings of the 33rd International
  Conference on Neural Information Processing Systems (NeurIPS)}}.
  \bibinfo{pages}{14584--14597}.
\newblock


\bibitem[\protect\citeauthoryear{Mathioudakis, Castillo, Barnabo, and
  Celis}{Mathioudakis et~al\mbox{.}}{2020}]%
        {Mathioudakis19a}
\bibfield{author}{\bibinfo{person}{Michael Mathioudakis},
  \bibinfo{person}{Carlos Castillo}, \bibinfo{person}{Giorgio Barnabo}, {and}
  \bibinfo{person}{Sergio Celis}.} \bibinfo{year}{2020}\natexlab{}.
\newblock \showarticletitle{Affirmative Action Policies for Top-k Candidates
  Selection, With an Application to the Design of Policies for University
  Admissions}. In \bibinfo{booktitle}{\emph{Proceedings of the ACM Symposium on
  Applied Computing (SAC)}}. \bibinfo{pages}{440--449}.
\newblock


\bibitem[\protect\citeauthoryear{O'Dea, Lagisz, Jennions, and Nakagawa}{O'Dea
  et~al\mbox{.}}{2018}]%
        {O'Dea18a}
\bibfield{author}{\bibinfo{person}{R.~E. O'Dea}, \bibinfo{person}{M. Lagisz},
  \bibinfo{person}{M.~D. Jennions}, {and} \bibinfo{person}{S. Nakagawa}.}
  \bibinfo{year}{2018}\natexlab{}.
\newblock \showarticletitle{Gender differences in individual variation in
  academic grades fail to fit expected patterns for STEM}.
\newblock \bibinfo{journal}{\emph{Nature Communications}} \bibinfo{volume}{9},
  \bibinfo{number}{1} (\bibinfo{year}{2018}), \bibinfo{pages}{3777}.
\newblock


\bibitem[\protect\citeauthoryear{Passariello}{Passariello}{2016}]%
        {Passariello16a}
\bibfield{author}{\bibinfo{person}{Christina Passariello}.}
  \bibinfo{year}{2016}\natexlab{}.
\newblock \showarticletitle{Tech Firms Borrow Football Play to Increase Hiring
  of Women}.
\newblock \bibinfo{journal}{\emph{Wall Street Journal}} (\bibinfo{date}{27
  Sept.} \bibinfo{year}{2016}).
\newblock


\bibitem[\protect\citeauthoryear{Pedreshi, Ruggieri, and Turini}{Pedreshi
  et~al\mbox{.}}{2008}]%
        {Pedreshi08a}
\bibfield{author}{\bibinfo{person}{Dino Pedreshi}, \bibinfo{person}{Salvatore
  Ruggieri}, {and} \bibinfo{person}{Franco Turini}.}
  \bibinfo{year}{2008}\natexlab{}.
\newblock \showarticletitle{Discrimination-aware Data Mining}. In
  \bibinfo{booktitle}{\emph{Proceedings of the 14th ACM SIGKDD International
  Conference on Knowledge Discovery and Data Mining (KDD)}}.
  \bibinfo{pages}{560--568}.
\newblock


\bibitem[\protect\citeauthoryear{Phelps}{Phelps}{1972}]%
        {phelps72}
\bibfield{author}{\bibinfo{person}{Edmund Phelps}.}
  \bibinfo{year}{1972}\natexlab{}.
\newblock \showarticletitle{The Statistical Theory of Racism and Sexism}.
\newblock \bibinfo{journal}{\emph{American Economic Review}}
  \bibinfo{volume}{62}, \bibinfo{number}{4} (\bibinfo{year}{1972}),
  \bibinfo{pages}{659--61}.
\newblock


\bibitem[\protect\citeauthoryear{Raghavan, Barocas, Kleinberg, and
  Levy}{Raghavan et~al\mbox{.}}{2020}]%
        {Raghavan20a}
\bibfield{author}{\bibinfo{person}{Manish Raghavan}, \bibinfo{person}{Solon
  Barocas}, \bibinfo{person}{Jon Kleinberg}, {and} \bibinfo{person}{Karen
  Levy}.} \bibinfo{year}{2020}\natexlab{}.
\newblock \showarticletitle{Mitigating Bias in Algorithmic Hiring: Evaluating
  Claims and Practices}. In \bibinfo{booktitle}{\emph{Proceedings of the 2020
  Conference on Fairness, Accountability, and Transparency (FAT*)}}.
  \bibinfo{pages}{469--481}.
\newblock


\bibitem[\protect\citeauthoryear{Zafar, Valera, Gomez~Rodriguez, and
  Gummadi}{Zafar et~al\mbox{.}}{2017a}]%
        {Zafar17a}
\bibfield{author}{\bibinfo{person}{Muhammad~B. Zafar}, \bibinfo{person}{Isabel
  Valera}, \bibinfo{person}{Manuel Gomez~Rodriguez}, {and}
  \bibinfo{person}{Krishna~P. Gummadi}.} \bibinfo{year}{2017}\natexlab{a}.
\newblock \showarticletitle{Fairness Beyond Disparate Treatment \& Disparate
  Impact: Learning Classification Without Disparate Mistreatment}. In
  \bibinfo{booktitle}{\emph{Proceedings of the 26th International Conference on
  World Wide Web (WWW)}}. \bibinfo{pages}{1171--1180}.
\newblock


\bibitem[\protect\citeauthoryear{Zafar, Valera, Gomez~Rogriguez, and
  Gummadi}{Zafar et~al\mbox{.}}{2017b}]%
        {Zafar17c}
\bibfield{author}{\bibinfo{person}{Muhammad~B. Zafar}, \bibinfo{person}{Isabel
  Valera}, \bibinfo{person}{Manuel Gomez~Rogriguez}, {and}
  \bibinfo{person}{Krishna~P. Gummadi}.} \bibinfo{year}{2017}\natexlab{b}.
\newblock \showarticletitle{{Fairness Constraints: Mechanisms for Fair
  Classification}}. In \bibinfo{booktitle}{\emph{Proceedings of the 20th
  International Conference on Artificial Intelligence and Statistics
  (AISTATS)}}. \bibinfo{pages}{962--970}.
\newblock


\bibitem[\protect\citeauthoryear{Zemel, Wu, Swersky, Pitassi, and Dwork}{Zemel
  et~al\mbox{.}}{2013}]%
        {zemel13}
\bibfield{author}{\bibinfo{person}{Rich Zemel}, \bibinfo{person}{Yu Wu},
  \bibinfo{person}{Kevin Swersky}, \bibinfo{person}{Toni Pitassi}, {and}
  \bibinfo{person}{Cynthia Dwork}.} \bibinfo{year}{2013}\natexlab{}.
\newblock \showarticletitle{Learning Fair Representations}. In
  \bibinfo{booktitle}{\emph{Proceedings of the 30th International Conference on
  Machine Learning (ICML)}}. \bibinfo{pages}{325--333}.
\newblock


\end{thebibliography}

\appendix

\section{Additional Plots}
\label{section: additional figures}

As was mentioned in Section~\ref{section: experiments}, we perform our study on different prior distributions of $\w$. In Section~\ref{section: experiments}, we studied the Pareto case. In this section, we show the selection results for other distributions: Uniform, Gaussian mixture and Beta distribution. 

The most interesting case is a two component Gaussian mixture, where with probability $\pi=0.2$ there appears a ``good'' candidate, and with probability $0.8$ appears a ``bad'' candidate. This is a typical situation where only small proportion of candidates are ``suitable'' for the selection. We set selection budget $\ay=0.1$, which means that we aim to select only candidates from the second peak. Parameters of this mixture are as follows: $\mu_1=0.0$, $\mu_2=0.2$, $\sigma_1 = \sigma_2=0.05$.

While performing \textbf{one-stage} selection, we vary the selection budget $\ax$ from 0.1 to 1. In \textbf{two-stage} selection, we fix $\ay=0.1$ and vary $\ax$ from $\ay$ to 1. The implicit variance for $B$-candidates is fixed to $\sxb=0.1$, the implicit variance for $A$-candidates varies with $k=1,2,3,4$ as $\sxa=k\sxb$. Parameters $\pa=0.4$ and $\pb=0.6$.

On top row of Figure~\ref{fig: selections differen priors}, we plot corresponding pdfs, on middle row we show the result for one-stage selection and on bottom row, the two-stage selection is displayed. We observe that the demographic parity selection outperforms the group oblivious algorithm in one-stage selection for small and large $\ax$. In two-stage selection, demographic parity is also better than the group oblivious algorithm for values $\ax$ close $\ay$, while for $\ax \gg \ay$, both tend to perform similarly.

\begin{figure} 

    \begin{subfigure}{0.3\textwidth}
        \includegraphics[width=\linewidth]{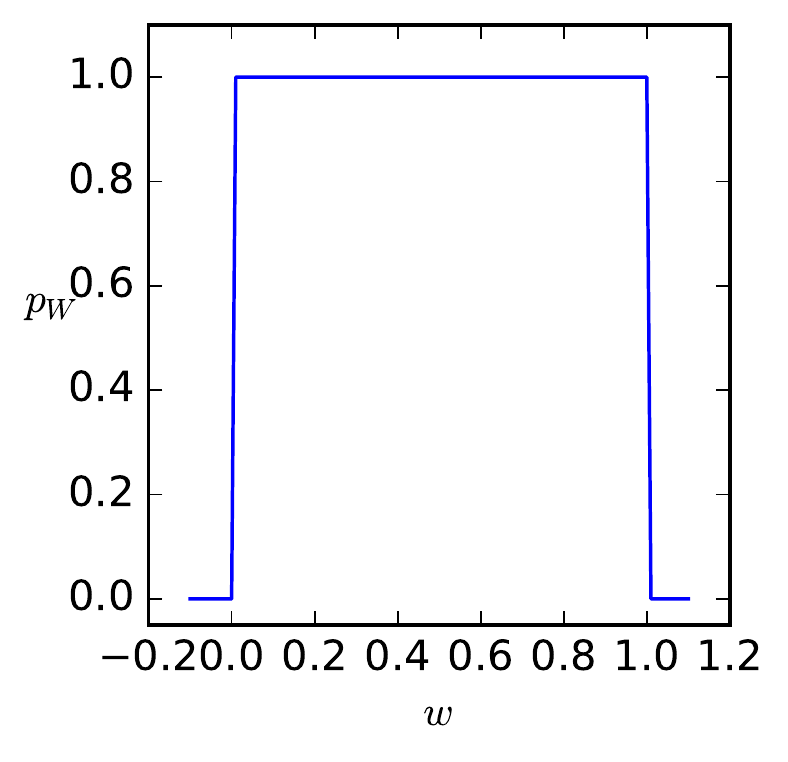}
        \end{subfigure}
        \begin{subfigure}{0.3\textwidth}
        \includegraphics[width=\linewidth]{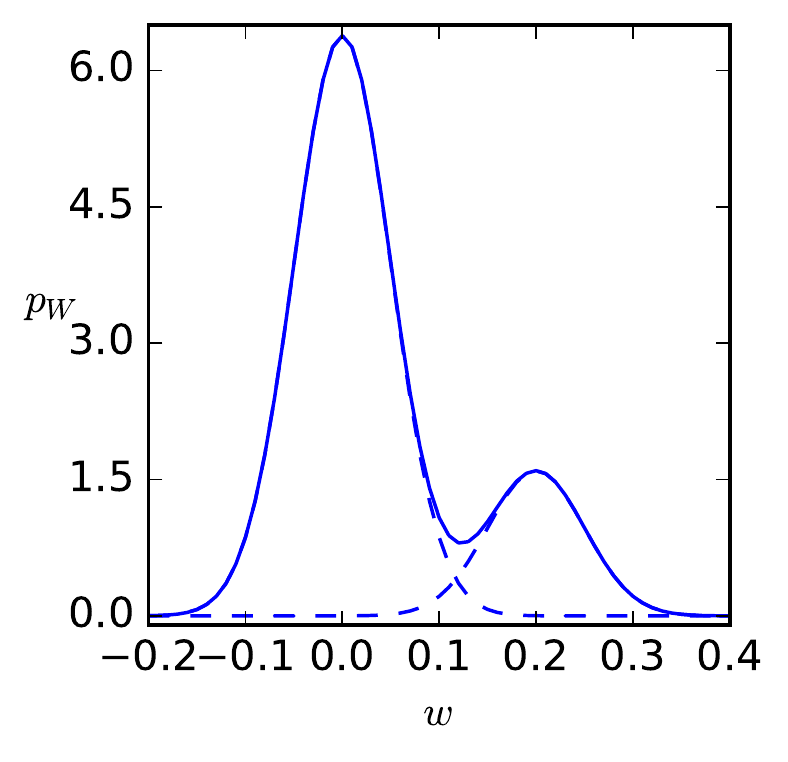}
        \end{subfigure}
        \begin{subfigure}{0.3\textwidth}
            \includegraphics[width=\linewidth]{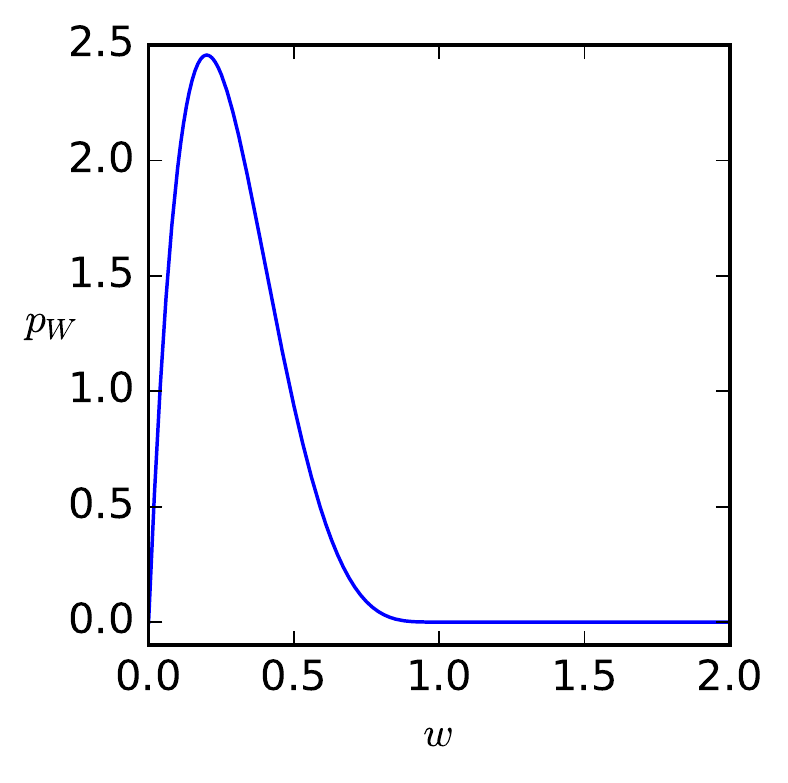}
        \end{subfigure}

    \medskip
    \begin{subfigure}{0.3\textwidth}
    \includegraphics[width=\linewidth]{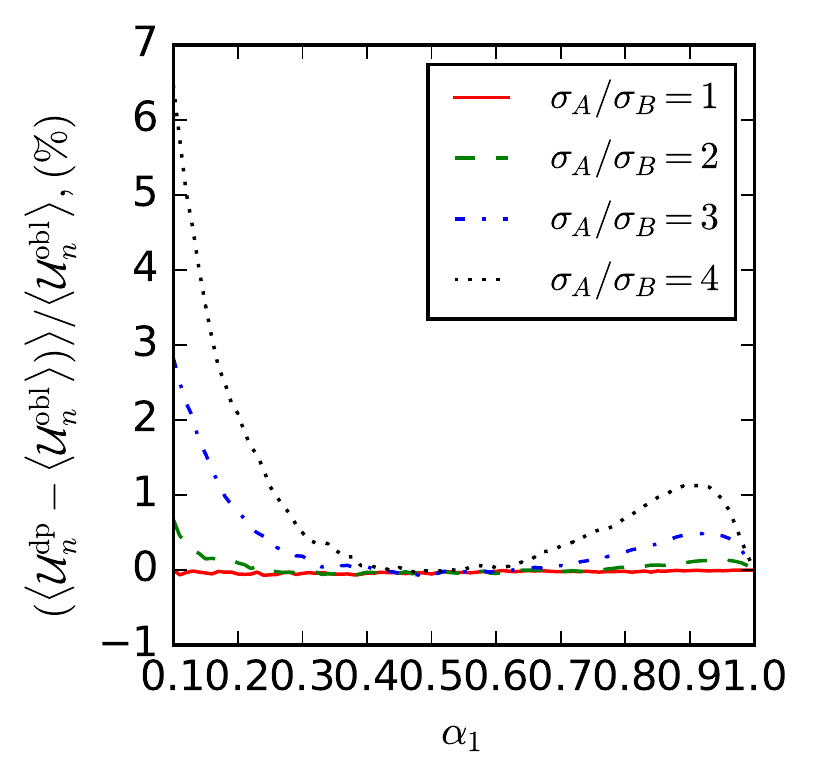}
    \end{subfigure}
    \begin{subfigure}{0.3\textwidth}
    \includegraphics[width=\linewidth]{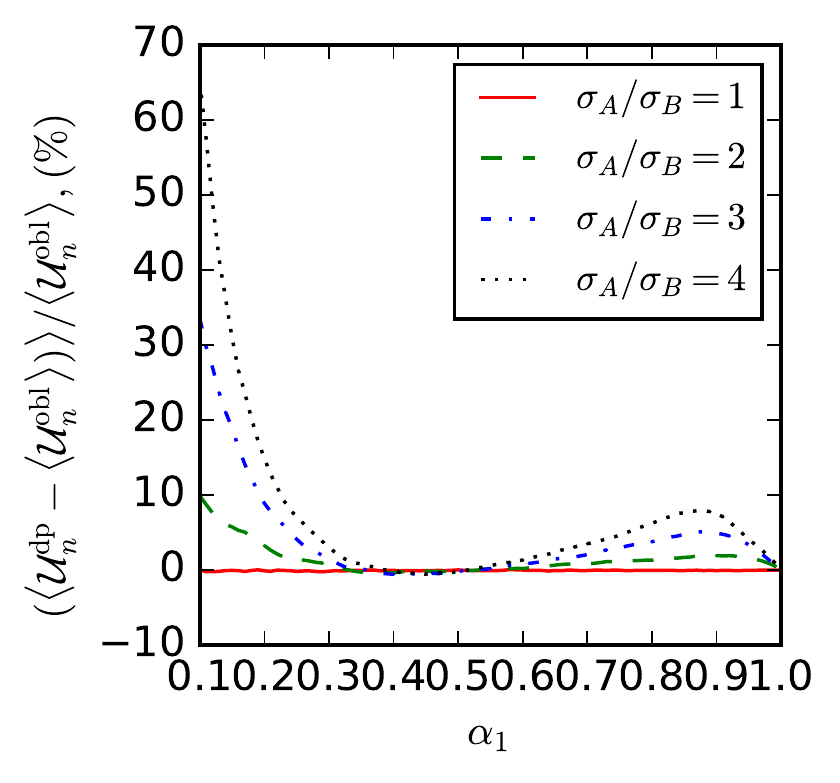}
    \end{subfigure}
    \begin{subfigure}{0.3\textwidth}
        \includegraphics[width=\linewidth]{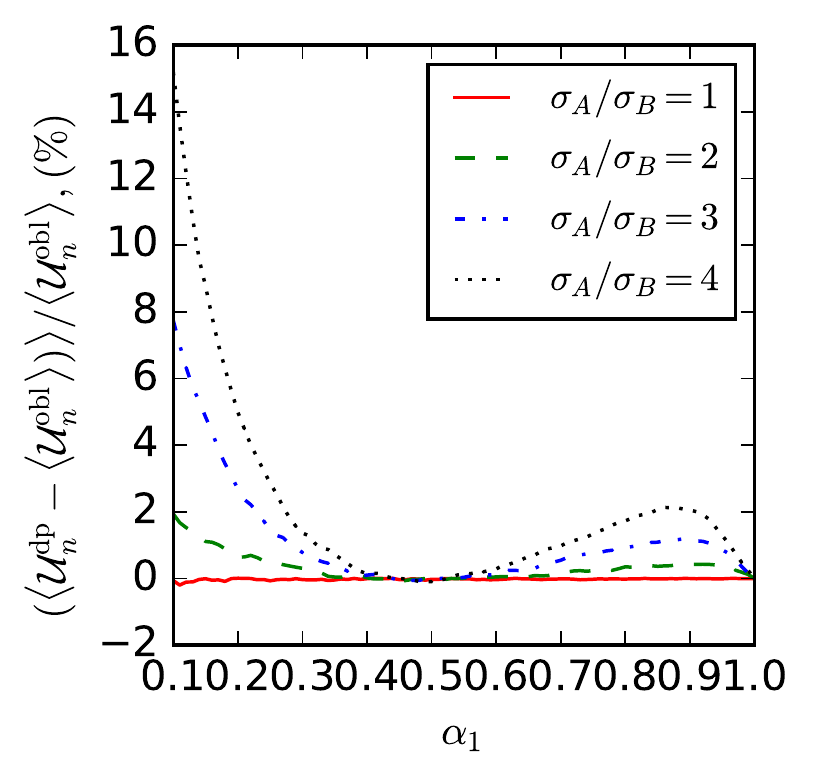}
    \end{subfigure}

    \medskip
    \begin{subfigure}{0.3\textwidth}
    \includegraphics[width=\linewidth]{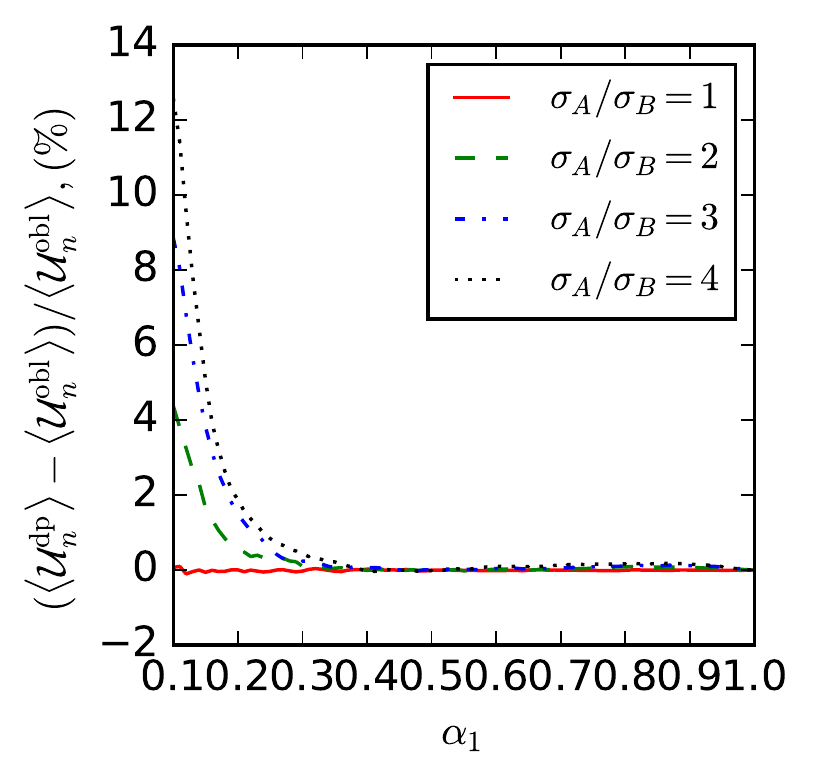}
    \caption{Uniform$(0, 1)$} \label{fig:c}
    \end{subfigure}
    \begin{subfigure}{0.3\textwidth}
    \includegraphics[width=\linewidth]{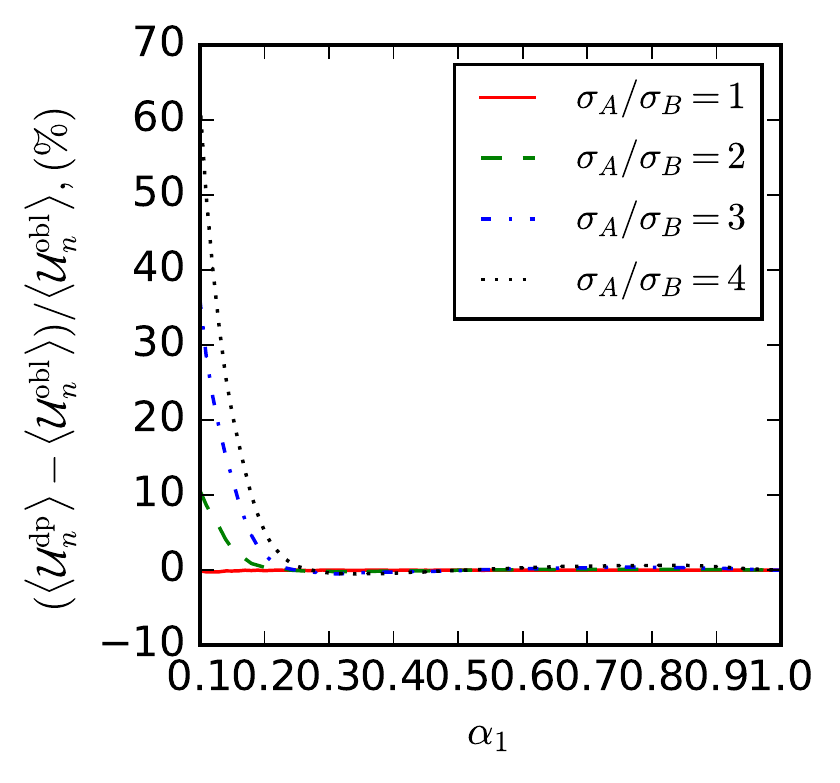}
    \caption{Gaussian Mixture} \label{fig:d}
    \end{subfigure}
    \begin{subfigure}{0.3\textwidth}
        \includegraphics[width=\linewidth]{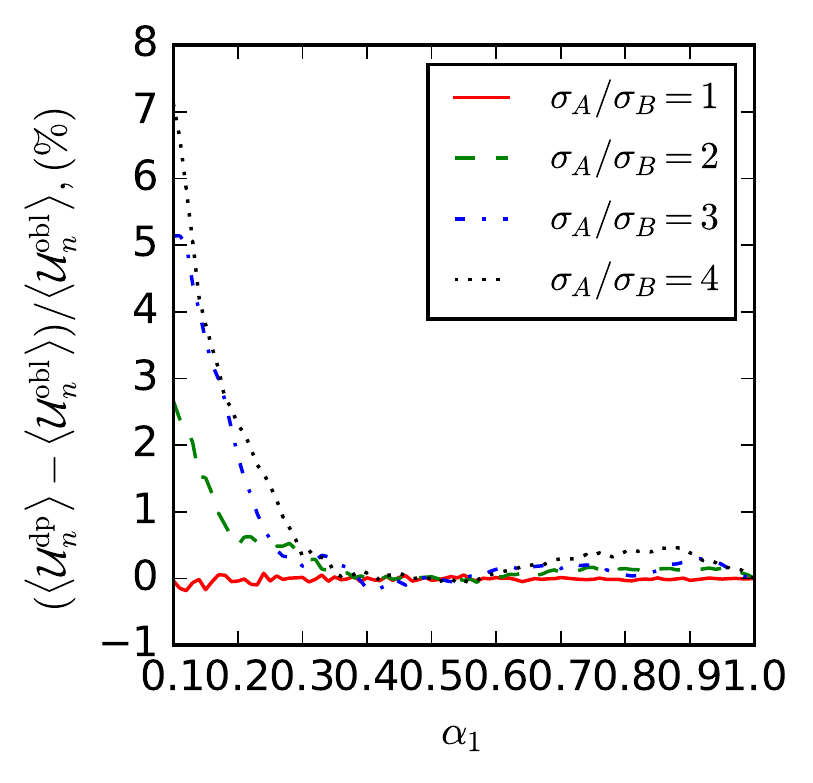}
        \caption{Beta$(2, 5)$}
         \label{fig:b}
        \end{subfigure}
    \caption{One- and two-stage selection results for different prior distributions of quality $\w$. The implicit variance for $B$-candidates is fixed to $\sxb=0.1$, the implicit variance for $A$-candidates varies with $k=1,2,3,4$ as $\sxa=k\sxb$. We set $n=1,000$ and average over $100$ experiments.}
    \label{fig: selections differen priors}
\end{figure}


\section{Ommited proofs}
\label{section: proofs}

In this section we provide detailed proofs of the statements given before. Namely, these are proofs of Lemma~\ref{lemma: selections for dp and greedy 1 stage}, Lemma~\ref{lemma: optimal selection} and Theorem~\ref{theorem: two stage}.

\subsection{Proof of Lemma \ref{lemma: selections for dp and greedy 1 stage}}
\label{proof: selections for dp and greedy 1 stage}

The distribution of $\w$ is common for both groups and follows a normal law with parameters $\mu_\w$ and $\sigma_\w^2$. The noise $\varepsilon$ is centered normal, thus,  $\wh=\w + \varepsilon\sigma_G$ follows a normal law with parameters $\mu_{\wh|G}= \mu_\w$ and $\sigma^2_{\wh|G} = \sq^2 + \sxg^2$. 

The selection  fraction $\pxggreedy$ is, by definition, the probability to observe a $G$-candidate with $\wh$ larger than $\txggreedy$ (where $\txggreedy = \txggreedy_A = \txggreedy_B$ since the thresholds are the same in the group-oblivious selection algoritm). Thus, $\pxggreedy = \Pb(\wh \geq \txggreedy \,|\,G) = \Phi^c\left(\frac{\txggreedy -\mq}{\sqrt{\sxg^2 + \sq^2}}\right)$. Recall that to shorten the notation, we write $\Pb(\wh \geq \txggreedy \,|\,G)$ for $\Pb(\wh_i \geq \txggreedy \,|\,G_i=G)$.

Now, we consider three different cases:
\begin{itemize}
\item[(1)]  if $\ax < 1/2$, then $\txggreedy > \mq$ and $\frac{\txggreedy -\mq}{\sqrt{\sxa^2 + \sq^2}} <  \frac{\txggreedy -\mq}{\sqrt{\sxb^2 + \sq^2}}$;
\item[(2)]  if $\ax = 1/2$, then $\txggreedy = \mq$ and $\frac{\txggreedy -\mq}{\sqrt{\sxa^2 + \sq^2}} =  \frac{\txggreedy -\mq}{\sqrt{\sxb^2 + \sq^2}} = 0$;
\item[(3)]  if $\ax > 1/2$, then $\txggreedy < \mq$ and $\frac{\txggreedy -\mq}{\sqrt{\sxa^2 + \sq^2}} >  \frac{\txggreedy -\mq}{\sqrt{\sxb^2 + \sq^2}}$. 
\end{itemize}
The function $\Phi^c$ is decreasing with its argument, thus the statement of lemma is a direct consequence of the above inequalities.

\subsection{Proof of Lemma \ref{lemma: optimal selection}}
\label{proof: optimal selection}

The following lemma is necessary to study the optimal one-stage selection. The proof of it is technical and we postpone it to Section~\ref{ssec:q' 1 stage proof}.
\begin{lemma}[Derivatives of One-Stage Utility $\Q$]
\label{lemma: q' 1 stage}
For  any $\ax$, where $0 < \ax \leq 1$, we have 
    \begin{enumerate}
    \item $\frac{d\Q}{d\pxa} = \frac{\pa}{\ax}\left[\frac{\txa \cdot \sq^2 + \mq \cdot \sxa^2}{\sq^2 + \sxa^2} - \frac{\txb \cdot \sq^2 + \mq \cdot \sxb^2}{\sq^2 + \sxb^2}\right]$;
    \item $\frac{d^2\Q}{d\pxa^2} < 0$,
\end{enumerate}
where $\twha$ and $\twhb$ are such that
\begin{align*}
\left\{\begin{array}{l}
    \Pb(\wh_A \ge \twha\,|\,G=A) = \pxa,\\
    \sum_{G\in\{A,B\}}\Pb(\wh \ge \twhg\,|\,G)\cdot\Pb(G) = \ax. \\
  \end{array}
\right. 
\end{align*}
\end{lemma}

The result in Lemma~\ref{lemma: q' 1 stage} helps to obtain several insights about the optimal one-stage selection. From the condition on a maximum $\frac{d\Q}{d\pxa}=0$, we obtain that optimal quantiles $\txgopt$ should satisfy $\txgopt = C\cdot(\sq^2 + \sxg^2) + \mq$ for some $C$ that does not depend on the group $G$. Since $\wh|G$ follows a normal law with parameters $\mu_{\wh|G} = \mq$,  $\sigma_{\wh|G}^2=\sxg^2 + \sq^2$,  the selection probability $\pxgopt$ can be written as $\pxgopt = \Pb(\wh \ge \txgopt\,|\, G) = \Phi^c\left(\frac{\txgopt - \mq}{\sqrt{\sq^2 + \sxg^2}}\right)$. Using the expression for optimal quantiles, we obtain $\pxgopt = \Phi^c\left(C\cdot\sqrt{\sq^2 + \sxg^2}\right)$.

Now, three cases are possible: 
\begin{itemize}
\item[(1)] If $\txgopt < \mq$, then $C < 0$. As a result, using the expression for $\pxgopt$, we  obtain that $\pxaopt > \pxbopt$.
\item[(2)] If $\txgopt > \mq$, then $C >  0$. Hence,  $\pxaopt < \pxbopt$.
\item[(3)] If $\txgopt = \mq$, then $C = 0$ and $\pxaopt = \pxbopt = 1/2$. 
\end{itemize}
Since $\pa \pxaopt + \pb\pxbopt  = \ax$, then $\pxadp$ is a convex combination of $\pxaopt$ and $\pxbopt$. Thus, $\pxadp < \pxaopt$ for $\txgopt < \mq$, $\pxadp > \pxaopt$ for $\txgopt > \mq$ and $\pxadp = \pxaopt$ for $\txgopt = \mq$.

\subsection{Proof of Theorem \ref{theorem: two stage}}
\label{proof: two stage}

The following lemma is necessary to study the optimal two-stage selection. Due to the technicality of its proof, we postpone it to Section~\ref{ssec:q' 2 stage proof}.
\begin{lemma}[Derivatives of Two-Stage Utility   $\Q$]
    \label{lemma: derivatives Q}
    For any $\ax$, $\ay$, where $0 < \ay \leq \ax$:
    \begin{enumerate}
    \item 
        $ \frac{d\Q}{d\pxa} = \frac{1}{\ay} \pa \left[ \hat \sigma_A \int_{-\infty}^{(\hat \mu_A - \ty)/
        \hat \sigma_{A}} \Phi(\tau)d\tau - \hat \sigma_B \int_{-\infty}^{(\hat \mu_B - \ty)/
        \hat \sigma_{B}} \Phi(\tau)d\tau  \right]$, where 
    \begin{align*}
        \hat \mu_G = \frac{\mq \sxg^2 + \twhg  \sq^2}{\sxg^2 + \sq^2},\;\; \hat \sigma_{G}^2 = \frac{\sxg^2  \sq^2}{\sxg^2  + \sq^2};
    \end{align*}
    \item $\frac{d^2\Q}{d\pxa^2} < 0$,
\end{enumerate}
where $\twha$, $\twhb$ and $\tw$ are such that
\begin{align*}
\left\{\begin{array}{l}
    \Pb(\wh_A \ge \twha\,|\,G=A) = \pxa,\\
    \sum_{G\in\{A,B\}}\Pb(\wh \ge \twhg\,|\,G)\cdot\Pb(G) = \ax, \\
    \sum_{G\in\{A,B\}}\Pb(\wh \ge \twhg, \w \ge \tw \,|\,G )\cdot\Pb(G) = \ay.
  \end{array}
\right. 
\end{align*}
\end{lemma}

The proof proceeds as follows. First, we show that for $\ax > 1/2$, we have $\pxaopt > \pxadp$. Then, using Lemma \ref{lemma: selections for dp and greedy 1 stage}, we have $\pxadp > \pxagreedy$ for $\ax > 1/2$. Using the property that $\pxagreedy < \pxadp < \pxaopt$, we conclude that there exists $\lambda \in (0,1)$, such that $\pxadp = \lambda \pxagreedy + (1-\lambda)\pxaopt$. By strict concavity of two-stage selection utility, we obtain
\begin{align*}
    \Q\left(\pxadp\right) &= \Q\left(\lambda \pxagreedy + (1-\lambda) \pxaopt\right) > \lambda\Q\left(\pxagreedy\right) + (1-\lambda) \Q\left(\pxaopt\right)  >  \Q\left(\pxagreedy\right).
\end{align*}
Hence, to complete the proof we need  to show that $\pxaopt > \pxadp$ for $\ax > 1/2$, which we do in the next few paragraphs.

The random variable $\wh|G$ follows a normal law with parameters $\mu_{\wh|G}=\mq$ and $\sigma^2_{\wh|G}=\sxg^2 + \sq^2$. Hence, the threshold $\txgdp = F^{-1}_{\wh|G}(1-\ax)$ can be simplified to $\txgdp = \mq + \sqrt{\sq^2 + \sxg^2} \cdot \Phi^{-1}(1-\ax)$. By using the expression for $\Q'$ from Lemma \ref{lemma: derivatives Q} and substituting $\txgdp$ into it, we obtain
\begin{align*}
\Q'(\pxadp) = \frac{1}{\ay} \pa \Big( \hat \sigma_A \int_{-\infty}^{(\hat \mu_A - \ty)/
            \hat \sigma_{A}} \Phi(\tau)d\tau - \hat \sigma_B \int_{-\infty}^{(\hat \mu_B - \ty)/
            \hat \sigma_{B}} \Phi(\tau)d\tau  \Big)\;,
\end{align*}
where $\hat \mu_G = \mq + \frac{\Phi^{-1}(1-\ax) \cdot\sq^2}{\sqrt{\sq^2 + \sxg^2}}$ and $\hat \sigma_G = \frac{\sxg^2  \sq^2}{\sxg^2 + \sq^2}$. As a result,
$\hat \mu_A - \ty > \hat \mu_B - \ty$ for $\Phi^{-1}(1-\ax) < 0 \iff \ax > 1/2$.
Thus,  
$\Q'(\pxadp) \geq I(\hat \sxa) - I(\hat \sxb)$, where $I(\sigma) = \sigma \int_{-\infty}^{I_0/\sigma} \Phi(\tau)d\tau$ for some $I_0 \in \mathbb{R}$.

We can show that the function $I(\sigma)$ is increasing for any value of $I_0$ by looking at its first derivative: 
\begin{align*}
I'_\sigma &= \int_{-\infty}^{I_0/\sigma} \Phi(\tau)d\tau + \sigma \Phi\left(\frac{I_0}{\sigma}\right) \left( -\frac{I_0}{\sigma^2}\right) 
= \frac{I_0}{\sigma} \Phi\left(\frac{I_0}{\sigma}\right) + \phi\left(\frac{I_0}{\sigma}\right) - \frac{I_0}{\sigma} \Phi\left(\frac{I_0}{\sigma}\right) 
= \phi\left(\frac{I_0}{\sigma}\right) > 0\;.
\end{align*}
Thus, for demographic parity selection, $\Q'(\pxadp) > 0$. Due to concavity of $\Q$, we obtain $\pxaopt > \pxadp$, which concludes the proof.

\subsection{Proof of Lemma \ref{lemma: q' 1 stage}}
\label{ssec:q' 1 stage proof}

In this section, we provide a proof of Lemma \ref{lemma: q' 1 stage}. We prove it in general setting, i.e. quality $\w$ follows any general distribution having a probability density function $p_\w(w)$.

The one-stage selection can be viewed as a two-stage selection where we select all candidates at second stage. This corresponds to setting the second-stage threshold to $\theta=-\infty$. Writing $\V(\twha, \twhb)=\V(\twha, \twhb,-\infty)$ by abuse of notation, the one-stage utility $\Q$ is defined as
\begin{align*}
    \Q(\pxa) &=\V(\twha, \twhb),\\
    &\text{with $\twha, \twhb$ such that}
      \left\{\begin{array}{l}
               \Pb(\wh \ge \twha\,|\,G=A) = \pxa,\\
               \sum_{G\in\{A,B\}}\Pb(\wh \ge \twhg\,|\,G)\cdot\Pb(G) = \ax.
             \end{array}
    \right.
    \nonumber
\end{align*}

For a general prior distribution of the quality $\w$, the quantities we study can be  expressed as
\begin{align*}
\V(\twha, \twhb) &= \frac 1 {\ax} \sum_G \pg  \int_{\txg} d\hat w \int dw \, w\cdot \pq \cdot \phig,\\
\Pb(\wh_G \ge \twhg\,|\,G)&=  \int_{\txg} d\hat w \int  dw\, \pq \cdot \phig=\pxg(\twhg),
\end{align*}
where $\pq$ is a probability density function of the distribution of quality $\w$. The symbol $\int$ here represents the integration over all support of the corresponding probability density function $\pq$, unless the limits are specified.
We prove the following statement about the derivatives of the one-stage utility $\Q$.
\begin{lemma*}[Derivatives of $\Q$]
For the one-stage selection utility $\Q$, we have
\begin{align}
\label{eq: Q1'}
&\frac{d\Q}{d\pxa} = \frac{p_A}{\ax} \left[ \frac{\int w \cdot \pq \cdot \phita dw}{\int  \pq \cdot \phita dw} - \frac{\int w \cdot \pq \cdot \phitb dw}{\int \pq \cdot \phitb dw} \right];\\
&\frac{d^2\Q}{d\pxa^2} < 0,
\end{align}    
where $\twha$ and $\twhb$ are such that
\begin{align*}
\left\{\begin{array}{l}
    \Pb(\wh_A \ge \twha\,|\,G=A) = \pxa,\\
    \sum_{G\in\{A,B\}}\Pb(\wh \ge \twhg\,|\,G)\cdot\Pb(G) = \ax. \\
  \end{array}
\right. 
\end{align*}
\end{lemma*}

\subsubsection*{First Derivative of $\Q$}
Using the parameterization by $\pxa$, we can write
\begin{align*}
\frac{d\Q}{d\pxa} &= \sum_G \frac{\partial \V}{\partial \twhg}\frac{d\twhg}{d\pxa}.
\end{align*}
From the first stage constraint $\pa\pxa + \pb\pxb = \ax$,
\begin{align*}
\pa \frac{d\pxa}{d\pxa} + \pb\frac{\partial\pxb}{\partial\twhb} \frac{d\twhb}{d\pxa} = 0 \implies \frac{d\twhb}{d\pxa} = -\frac{\pa}{\pb} \frac{\partial \twhb}{\partial \pxb}.
\end{align*}
Then $\frac{d\Q}{d\pxa} = \pa \left( \frac{\partial \V}{\partial \twha}\frac{\partial\twha}{\partial\pxa}- \frac{\partial \V}{\partial \twhb}\frac{\partial\twhb}{\partial\pxb} \right)$ from which \eqref{eq: Q1'}  follows directly.

\subsubsection*{Second Derivative of $\Q$}

The second derivative of $\Q$ with respect to $\pxa$, then can be written as $\frac{d^2\Q}{d\pxa^2}  = \sum_G \frac{\partial \Z}{\partial \twhg}\frac{d\twhg}{d\pxa}$, where by $\Z$ we denote the expression \eqref{eq: Q1'}.
Using \eqref{eq: Q1'}, we conclude that
\begin{align*}
\frac{\partial \Z}{\partial \twhg} &= \frac{-\int w \cdot \pq \cdot \left(\frac{\twhg - w}{\sxg^2} \right)\phitg\,dw \int \pq \phitg\,dw}{\ax\left(\int \pq \phitg\,dw\right)^2} \\
& \phantom{=} + \frac{\int w\cdot\pq\phitg\,dw \int \pq \cdot\left(\frac{\twhg - w}{\sxg^2}\right)\phitg\,dw}{\ax\left(\int \pq \phitg\,dw\right)^2}.
\end{align*}

If we use the notation $dP = \pq \phitg dw$, then
\begin{align*}
\frac{\partial \Z}{\partial \twhg} 
&= \frac{-\int w (\twhg - w)dP \cdot \int dP + \int w dP \cdot \int (\twhg - w)dP}{\ax\sxg^2 \left(\int dP\right)^2}\\
& = \frac{\int w^2  dP \cdot \int dP - \int w  dP \cdot \int wdP}{\ax\sxg^2\left(\int dP\right)^2}
= \frac{1}{\ax\sxg^2} \left[ \frac{\int w^2 dP}{\int dP}  - \frac{\int w dP}{\int dP} \cdot \frac{\int w dP}{\int dP}\right].
\end{align*}
Since $f(w) = w^2$ and $g(w) = w$ are increasing functions of $w$, then by applying Harris inequality \cite{fortuin71} to $f$, $g$ and probability measure $d\mu = dP / \int dP$, we obtain
\begin{align*}
\frac{\int w^2 dP}{\int dP}  - \frac{\int w dP}{\int dP} \cdot \frac{\int w dP}{\int dP} > 0\;.
\end{align*}
Hence, the second derivative of $\Q$ with respect to $\pxa$ can be written as
\begin{align*}
\frac{d^2\Q}{d\pxa^2} &= \pa \frac{d\twha}{d\pxa} \frac{\partial \Z}{\partial \twha}   - \pa \frac{d\twhb}{\pxa}\frac{\partial \Z}{\partial \twhb} = \pa \frac{\partial\twha}{\partial \pxa} \frac{\partial \Z}{\partial \twha}   + \frac{\pa^2}{\pb} \frac{\partial \twhb}{\partial \pxb} \frac{\partial \Z}{\partial \twhb} < 0.
\end{align*}

\subsection{Proof of Lemma \ref{lemma: derivatives Q}}
\label{ssec:q' 2 stage proof}

In this section, we provide a proof of Lemma~\ref{lemma: derivatives Q}. We also prove it in general setting, i.e. quality $\w$ follows any general distribution having a probability density function $p_\w(w)$.

The two-stage utility  is
\begin{align*}
    \Q(\pxa) &=\V(\twha, \twhb, \tw),\\
    &\text{with $\twha, \twhb, \tw$ such that}
      \left\{\begin{array}{l}
               \Pb(\wh_A \ge \twha\,|\,G=A) = \pxa,\\
               \sum_{G\in\{A,B\}}\Pb(\wh \ge \twhg\,|\,G)\cdot\Pb(G) = \ax, \\
               \sum_{G\in\{A,B\}}\Pb(\wh \ge \twhg, \w \ge \tw \,|\,G )\cdot\Pb(G) = \ay.
             \end{array}
    \right.
    \nonumber
\end{align*}

For a general prior distribution of the quality $\w$, the quantities we study can be written as
\begin{align*}
\V(\twha, \twhb, \tw) &= \frac 1 {\ay} \sum_G \pg  \int_{\twhg} d\hat w \int_{\ty} dw  \,w\cdot \pq \cdot \phig,\\
\Pb(\wh \ge \twhg\,|\,G) &=  \int_{\twhg} d \hat w \int  dw \, \pq \cdot \phig = \pxg(\twhg), \\
\Pb(\wh \ge \twhg, \w \ge \tw \,|\,G )&=  \int_{\twhg} d\hat w\int_{\ty}  dw\,\pq \cdot \phig =  \pyg(\twhg,\tw),
\end{align*}
where $\pq$ is a probability density function of the distribution of the quality $\w$. The symbol $\int$ here represents the integration over all support of the corresponding probability density function $\pq$, unless the limits are specified. We want to prove the following statement about the derivatives of the two-stage utility $\Q$.
\begin{lemma*}[Derivatives of $\Q$]
For the two-stage utility $\Q$, we have
\begin{align}
\label{eq: general Q'}
&\frac{d\Q}{d\pxa} = \frac{p_A}{\ay} \left( \frac{\int_{\ty} (w - \ty)\cdot \pq \cdot \phita dw}{\int  \pq \cdot \phita dw} - \frac{\int_{\ty} (w - \ty)\cdot \pq \cdot \phitb dw}{\int \pq \cdot \phitb dw} \right);\\
&\frac{d^2\Q}{d\pxa^2} <0,
\end{align}  
where $\twha$, $\twhb$ and $\tw$ are such that
\begin{align*}
\left\{\begin{array}{l}
    \Pb(\wh_A \ge \twha\,|\,G=A) = \pxa,\\
    \sum_{G\in\{A,B\}}\Pb(\wh \ge \twhg\,|\,G)\cdot\Pb(G) = \ax, \\
    \sum_{G\in\{A,B\}}\Pb(\wh \ge \twhg, \w \ge \tw \,|\,G )\cdot\Pb(G) = \ay.
  \end{array}
\right. 
\end{align*}
\end{lemma*}  

\subsubsection*{First Derivative of $\Q$}

Using the parameterization by $\pxa$, $\frac{d\Q}{d\pxa}$ can be written as
\begin{align*}
\frac{d\Q}{d\pxa} &= \sum_G \frac{\partial \V}{\partial \twhg}\frac{d\twhg}{d\pxa} + \frac{\partial \V}{\partial \ty}\frac{d\ty}{d\pxa}.
\end{align*}
From the first stage constraint $\pa\pxa+ \pb\pxb = \ax$:
\begin{align*}
\pa \frac{d\pxa}{d\pxa} + \pb\frac{\partial\pxb}{\partial\twhb} \frac{d\twhb}{d\pxa} = 0 \implies \frac{d\twhb}{d\pxa} = -\frac{\pa}{\pb} \frac{\partial \twhb}{\partial \pxb}.
\end{align*}
From second stage budget constraint $\py:= \pa\pya + \pb\pyb = \ay$, we get
\begin{align*}
 \sum_G \pg \frac{\partial \pyg}{\partial \twhg}\frac{d\twhg}{d\pxa} + \frac{\partial \py}{\partial \ty} \frac{d\ty}{d\pxa} = 0
\implies \frac{d\ty}{d\pxa} = - \left(\frac{\partial \py}{\partial \ty}\right)^{-1}\sum_G \pg  \frac{\partial \pyg}{\partial \twhg}\frac{d\twhg}{d\pxa}.
\end{align*}
Hence,
\begin{align*}
\frac{d\Q}{d\pxa} &= \pa \left( \frac{\partial \V}{\partial \twha}\frac{\partial\twha}{\partial\pxa}- \frac{\partial \V}{\partial \twhb}\frac{\partial\twhb}{\partial\pxb} \right)
- \left(\frac{\partial \py}{\partial \ty}\right)^{-1} \pa \left(\frac{\partial \pya}{\partial \twha}\frac{\partial \twha}{\partial \pxa} - \frac{\partial \pyb}{\partial \twhb}\frac{\partial \twhb}{\partial \pxb}\right) \frac{\partial \V}{\partial \ty}\\
& = \pa \left[ \frac{\partial \twha}{\partial \pxa}\left(\frac{\partial \V}{\partial \twha} - \frac{\partial \V}{\partial \ty} \frac{\partial \ty}{\partial \py} \frac{\partial \pya}{\partial \txa} \right) - \frac{\partial \twhb}{\partial \pxb} \left( \frac{\partial \V}{\partial \twhb} - \frac{\partial \V}{\partial \ty} \frac{\partial \ty}{\partial \py} \frac{\partial \pyb}{\partial \twhb}  \right) \right].
\end{align*}

Let us consider the following quantity:
\begin{align*}
\frac{\partial \twhg}{\partial \pxg} &\left( \frac{\partial \V}{\partial \twhg} - \frac{\partial \V}{\partial \ty} \frac{\partial \ty}{\partial \py} \frac{\partial \pyg}{\partial \twhg} \right)  = 
\frac{\frac 1 {\ay} \int_{\ty} w  \pq \frac 1 {\sxg}\phi\left(\frac{\twhg - w}{\sxg}\right)\,dw }
{\int \pq \frac{1}{\sxg}\phi\left(\frac{\twhg - w}{\sxg}\right)\,dw}\\
& - \frac{\int_{\ty} \pq \frac 1 {\sigma_g}\phi\left(\frac{\twhg - w}{\sxg}\right)\, dw}
{\int \pq \frac{1}{\sxg}\phi\left(\frac{\twhg - w}{\sxg}\right)\,dw}
\cdot \underbrace{\frac{\frac 1 {\ay} \sum_g \pg  \int_{\twhg} \ty  p_W(\ty) \frac{1}{\sxg}\phi\left(\frac{\hat w - \ty}{\sxg}\right) \, d\hat w}
{\sum_g \pg \int_{\twhg} p_W(\ty) \frac{1}{\sxg}\phi\left(\frac{\hat w - \ty}{\sxg}\right)\, d\hat w}}_{=\ty}\\
& = \frac{1}{\ay} \frac{\int_{\ty} (w - \ty)\cdot \pq \cdot \phitg dw}{\int  \pq \cdot \phitg dw}.
\end{align*}
Equation \eqref{eq: general Q'}  follows directly from  the expression developed above.

\subsubsection*{Second Derivative of $\Q$}

The second derivative of $\Q$ with respect to $\pxa$, then can be written as:
\begin{align*}
\frac{d^2\Q}{d\pxa^2}  = \sum_G \frac{\partial \Z}{\partial \twhg}\frac{d\twhg}{d\pxa} + \frac{\partial \Z}{\partial \ty}\frac{d\ty}{d\pxa},
\end{align*}
where by $\Z$ we denote the expression \eqref{eq: general Q'}. Then, by using \eqref{eq: general Q'}, we can calculate
\begin{align*}
\frac{\partial \Z}{\partial \ty} &= 0;\\
\frac{\partial \Z}{\partial \twhg} &= \frac{-\int_{\ty} (w-\ty) \cdot \pq \cdot \left(\frac{\twhg - w}{\sxg^2} \right)\phitg\,dw \int \pq \phitg\,dw}{\ay\left(\int \pq \phitg\,dw\right)^2}\\
& \phantom{=} + \frac{\int_{\ty} (w-\ty)\cdot\pq\phitg\,dw \int \pq \cdot\left(\frac{\twhg -w}{\sxg^2}\right)\phitg\,dw}{\ay\left(\int \pq \phitg\,dw\right)^2}.\\
\end{align*}

Using the notation $dP = \pq \phitg dw$, we get
\begin{align*}
\frac{\partial \Z}{\partial \twhg} 
&= \frac{-\int_{\ty} (w-\ty)(\twhg - w)dP \cdot \int dP + \int_{\ty} (w-\ty) dP \cdot \int (\twhg - w)dP}{\ay\sxg^2 \left(\int dP\right)^2}\\
& = \frac{\int_{\ty}(w^2 - \ty w)  dP \cdot \int dP - \int_{\ty} (w-\ty) dP \cdot \int wdP}{\ay\sxg^2\left(\int dP\right)^2}\\
& = \frac{1}{\ay\sxg^2} \left[ \frac{\int w \cdot (w-\ty) \1_{w \geq \ty} dP}{\int dP}  - \frac{\int (w-\ty) \1_{w \geq \ty} dP}{\int dP} \cdot \frac{\int w dP}{\int dP}\right].\\
\end{align*}
Since $f(w) = w$ and $g(w) = (w-\ty)\cdot \1_{w\geq \ty}$ are increasing functions of $w$, then by applying Harris inequality \cite{fortuin71} to $f$, $g$ and probability measure $d\mu = dP / \int dP$, we obtain:
\begin{align*}
\frac{\int w \cdot (w-\ty) \1_{w \geq \ty} dP}{\int dP}  - \frac{\int (w-\ty) \1_{w \geq \ty} dP}{\int dP} \cdot \frac{\int w dP}{\int dP} > 0\;.
\end{align*}
Then, the second derivative of $\Q$ with respect to $\pxa$ can be written as
\begin{align*}
\frac{d^2\Q}{d\pxa^2}  &= \pa \frac{d\twha}{d\pxa} \frac{\partial \Z}{\partial \twha}   - \pa \frac{d\twhb}{\pxa}\frac{\partial \Z}{\partial \twhb}
= \pa \frac{\partial\twha}{\partial \pxa} \frac{\partial \Z}{\partial \twha}   + \frac{\pa^2}{\pb} \frac{\partial \twhb}{\partial \pxb} \frac{\partial \Z}{\partial \twhb} < 0.
\end{align*}


\end{document}